\documentclass{article}

\usepackage{arxiv}

\usepackage[utf8]{inputenc} 
\usepackage[T1]{fontenc}    
\usepackage{hyperref}       
\usepackage{url}            
\usepackage{booktabs}       
\usepackage{amsfonts}       
\usepackage{nicefrac}       
\usepackage{microtype}      
\usepackage{lipsum}
\usepackage{graphicx}
\graphicspath{ {./images/} }

\usepackage{amsmath}
\usepackage{amssymb}
\usepackage{mathtools}
\usepackage{amsthm}
\usepackage{commath}

\usepackage{subcaption}
\usepackage{paralist}
\usepackage{tabularx}
\usepackage{marvosym}
\usepackage{xcolor}
\usepackage{ulem}
\usepackage{svg}
\usepackage{tikz}
\usepackage{cite}

\usepackage{pgfplots}
\pgfplotsset{compat=1.18}

\usepackage{siunitx}
\usepackage{hyperref}
\usepackage{cleveref}

\crefformat{footnote}{#2\footnotemark[#1]#3}

\theoremstyle{plain}
\newtheorem{theorem}{Theorem}[section]

\newtheorem{lemma}[theorem]{Lemma}

\theoremstyle{definition}
\newtheorem{definition}[theorem]{Definition}
\newtheorem{assumption}[theorem]{Assumption}
\theoremstyle{remark}

\newcommand{\bigO}{\mathcal{O}}

\newcommand{\Var}{\mathrm{Var}}
\newcommand\given[1][]{\:#1\vert\:}

\title{Boltzmann Price: Toward Understanding the Fair Price in High-Frequency Markets}

\author{
 Przemys\l aw Rola \\
  Cracow University of Economics\\
  Department of Mathematics\\
  30-427 Krakow, ul. Rakowicka 27 \\
  \texttt{przemyslaw.rola@outlook.com} \\
}

\begin{document}
\maketitle
\begin{abstract}
In this paper, we introduce a parametrized family of prices derived from the Maximum Entropy Principle. The price is obtained from the distribution that minimizes bias, given the bid and ask volume imbalance at the top of the order book. Under specific parameter choices, it closely approximates the mid-price or the weighted mid-price. Using probabilities of bid and ask states, we propose a model of price dynamics in which both drift and volatility are driven by volume imbalance. Compared to standard models like Bachelier or Geometric Brownian Motion with constant volatility, our model can generate higher kurtosis and heavy-tailed distributions. Additionally, the drift term naturally emerges as a consequence of the order book imbalance. We validate the model through simulation and demonstrate its fit to historical equity data. The model provides a theoretical framework, integrating price, volume imbalance, and spread.
\end{abstract}


\section{Introduction}
\label{sec:intro}

High-frequency trading (HFT) has transformed modern financial markets, leveraging advanced computational techniques and low-latency infrastructure to execute trades at unprecedented speeds. In electronic markets, trading occurs through the interaction of limit order books (LOBs), where participants post limit or market orders to provide or consume liquidity. As presented in Figure~\ref{fig:lob_example}, the LOB represents a dynamic snapshot of supply and demand, with the bid-ask spread serving as a key indicator of market liquidity. The spread represents the cost of an instantaneous round-trip transaction of one share, i.e., buying and immediately selling one share, or vice versa.

For a trade to take place, an agent must consume liquidity by placing a market order to buy or sell a specific number of shares, executing at the best available price if the order book's volume at that price can fulfill the order. If insufficient volume exists, the price moves up (or down) through the order book's price levels until the order is completely filled.


As observed, for example, in \cite{Bonart_2018}, price changes are mostly driven by liquidity providers who agree on a hidden fundamental price. In this study, we introduce the fundamental price, as a consequence of the Maximum Entropy Principle applied to the two-state system representing the bid and ask quotes, where the probability of each state depends on volume imbalance. It can be interpreted that market participants aim to minimize bias arising from incomplete information about volume imbalances. According to the Maximum Entropy Principle, when information about a probability distribution is limited, the least biased (or most unbiased) distribution should be selected to avoid introducing unjustified assumptions.

In many market microstructure studies, there exists an unobservable reference price, often called the fundamental price, efficient price or fair price.\cite{Jaisson_2015} In practice, two main candidates for this price are the mid-price and the weighted mid-price. However, the mid-price is a low frequency signal, as it does not account for bid and ask volumes and exhibits high autocorrelation. Conversely, the weighted mid-price is noisy, fluctuating every update of the top of the book (TOB), and lacks a theoretical justification as a fair price estimator. In this study, we aim to address some of these limitations by proposing a measure that is less noisy than the weighted mid-price, incorporates volume imbalance, and provides a theoretical justification.

\begin{figure}[ht]
    \centering
    \begin{tikzpicture}
    \begin{axis}[
        xlabel={Price},
        ylabel={Volume},
        legend pos=north west,
        legend style={
        draw=black, 
        fill=white, 
        fill opacity=0.8, 
        text opacity=1, 
        font=\small\sffamily 
    },
        ybar,
        bar width=0.5cm,
        symbolic x coords={A, B, $P^b$, $P^a$, E, F},
        xtick={$P^b$, $P^a$},
        ytick={150, 200},
        yticklabels={$Q^a$, $Q^b$},
        enlarge x limits=0.3,
        tick style={draw=none}
    ]

    \addplot[blue, fill=blue!50] coordinates {
        (A, 100)
        (B, 150)
        ($P^b$, 200)
    };
    \addlegendentry{Bid}

    \addplot[red, fill=red!50] coordinates {
        ($P^a$, 150)
        (E, 175)
        (F, 250)
    };
    \addlegendentry{Ask}

    \draw[dashed] (axis description cs:0,0.64) -- (axis cs:$P^b$,200); 
    \draw[dashed] (axis description cs:0,0.36) -- (axis cs:$P^a$,150); 


    \end{axis}
    
    \end{tikzpicture}
    \caption{Symbolic representation of a Limit Order Book.}
    \label{fig:lob_example}
\end{figure}
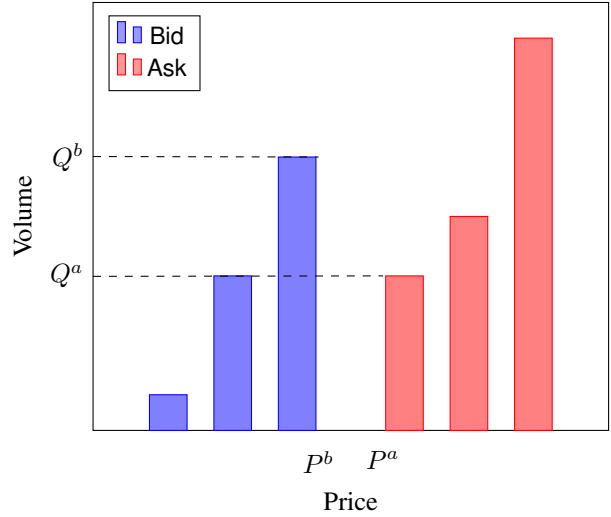


The (quoted) spread is defined as:
\begin{equation}
   \textit{spread} = S = P^a - P^b.
\end{equation}

Furthermore, given the best bid size $Q^b$ and the best ask size $Q^a$, the bid and ask (TOB) volume imbalances are defined as:\footnote{Usually $q$ is called in literature volume or order book imbalance. We use names bid and ask imbalance for clarity. Many authors, for example, \cite{Stoikov_2018, HAGSTROMER2021} denote it $I$ instead of $q$ in this paper.}
\begin{equation}
    q^b = q = \frac{Q^b}{Q^b + Q^a}, \qquad \qquad q^a = 1 - q = \frac{Q^a}{Q^b + Q^a}.
\end{equation}

where $q$ denotes the bid imbalance, and $1-q$ represents the ask imbalance.

We will use the (centered) imbalance given by:\footnote{Throughout this paper, `imbalance' is defined as $q$. When referring to centered imbalance, we will explicitly denote it as $\theta$.}
\begin{equation}
    \theta = q^b - \frac{1}{2} = \frac{Q^b - Q^a}{2\left( Q^b + Q^a \right)}.
\end{equation}

Sometimes in the literature \cite{Lipton2013}, the book imbalance is referred to as
\begin{equation}
    \mathcal{I} = 2\theta = \frac{Q^b - Q^a}{Q^b + Q^a}.
\end{equation}

One of the most widely used measures of fair price in market microstructure is the mid-price, defined as:
\begin{equation}
    P^{\text{mid}} = \frac{1}{2} \left( P^b + P^a \right),
\end{equation}

where $P^b$ represents the best bid price and $P^a$ denotes the best ask price. As mentioned, the mid-price is a low-frequency signal that does not account for volume imbalances in the limit order book.

To address this limitation, an alternative measure, the weighted mid-price, is often employed. This measure incorporates top of the book (TOB) volume imbalances to provide a more accurate estimate of fair price, typically expressed as
\begin{equation}
    P^{w} = \frac{Q^b P^a + Q^a P^b}{Q^b + Q^a} = q^b P^a + q^a P^b.
\end{equation}

Recently, some researchers have proposed alternatives to existing microstructure prices, including \cite{Bonart_2018, Jaisson_2015, Lehalle_2017}. Of particular interest in this study is the micro-price proposed by Stoikov\cite{Stoikov_2018}. This price can be described as a mid-price adjusted to account for both the bid-ask spread and volume imbalance.

As discussed, for example, in \cite{Lipton2013,HAGSTROMER2021} the dynamics of the TOB reflects the general intention of the market. As such, the order book imbalance can be interpreted as a valuable signal for the next price movement and could be picked up by traders.\footnote{In practice it could be used by the appropriate algorithmic trading strategy.} It is then intuitive what is stated in \cite{Lipton2013}: "When the number of shares available at the bid exceeds those at the ask, participants expect the next price movement to be upwards, and inversely, for the ask."\footnote{To what extent it is a valuable signal should depend on the book imbalance. The higher the imbalance, the more valuable the signal should be.}

\subsection{Contributions}

In this study, we introduce a novel price---the Boltzmann price---derived using the Maximum Entropy Principle. We analyze its properties and demonstrate its relationship to existing benchmarks, such as the mid-price and the weighted mid-price.

We further derive the dynamics of a price process, in which both drift and volatility are influenced by volume imbalance. Our results show that this model naturally generates higher kurtosis and heavier tails compared to, for example, the Bachelier \cite{Bachelier_1900} or Geometric Brownian Motion \cite{Black_1973} models.

Using historical market data, we show that the proposed model fits the distribution of price changes well. It provides a theoretical framework that integrates price dynamics, volume imbalance, and spread behavior.

\section{Related Work}

High-frequency trading (HFT) and market microstructure are crucial areas of financial research, with some attention focused on estimating fair asset prices using limit order book (LOB) data. \cite{Stoikov_2018} introduced the micro-price, a high-frequency estimator of future prices that incorporates the bid-ask spread
and the top of the book (TOB) imbalance. Constructed as a martingale, the micro-price provides a robust alternative to traditional metrics, such as the mid-price or volume-weighted mid-price, offering potentially greater predictive accuracy for HFT strategies. Complementing this, \cite{HAGSTROMER2021} examined biases in
effective bid-ask spread estimation, finding that spreads measured relative to the spread midpoint overstate the true effective spread by 13\% to 18\% for S\&P~500 stocks. This suggests that weighted mid-prices or micro-prices may serve as more reliable proxies for liquidity assessment than the traditional midpoint.
Consequently, accurate definition and estimation of fair prices are critical for various applications in market microstructure and trading.

Some researchers have proposed alternative definitions of a microstructure price that is directly
related to the bid and ask quotes. There is a great deal of literature devoted to the fundamental price \cite{Madhavan_1997, Bonart_2018}, efficient price \cite{Robert_2011, Delattre_2013} and fair price \cite{Gueant_2013b, Jaisson_2015, Lehalle_2017, Stoikov_2018}, with some alternatives to the mid-price and the weighted-mid price proposed.

Shannon entropy has been applied in several studies, for example, to test the Efficient Market Hypothesis \cite{Zhang_1999, Shternshis_2022} and to model the information available to informed traders \cite{Touzo_2021}. In the latter, the authors established a general analogy between the Glosten–Milgrom model \cite{Glosten_1985} and the Szil\'{a}rd information engine. Their analysis considered a simplified market in which a stock can take on only two values: $1$ with probability $p$, or $0$ with probability $1 - p$. To define \textit{market temperature}---a measure of noise in the market---they employed the Boltzmann distribution.

The Boltzmann price introduced in this paper recalls the idea that appeared in Prediction Markets (PMs). However, in PMs, the price of a security $i$ is derived as the partial derivative of the cost function with respect to its quantity $q_i$. This price represents the cost of purchasing an infinitesimal quantity of a share. The security $i$ pays $1$ unit (e.g., dollar) if outcome $i$ occurs, and its price reflects the estimated probability of that outcome. For mutually exclusive and exhaustive outcomes, the prices associated with these outcomes sum up to 1. One of the most widely used automated market makers (AMMs) in Internet PMs is the Logarithmic Market Scoring Rule (LMSR) \cite{Hanson_2007} and its variant, the Liquidity Sensitive Logarithmic Market Scoring Rule (LS-LMSR) \cite{Othman_2013}. There is a great deal of literature on PMs, including \cite{Hanson_2003, Wolfers_2004, Chen_2007, Lekwijit_2018,Carvalho_2023,Bossaerts_2024}.

Fair price discovery is important in terms of market impact and execution costs. Relevant papers on this topic include, but are not limited to, \cite{Almgren_1999, Almgren_2001, Almgren_2003, Gabaix_2003, Huberman_2004, Almgren_2005, Wyart_2008, Bouchaud_2009, Gatheral_2010, Eisler_2012, Hautsch_2012, Gueant_2013a, Farmer_2013, Bacry_2014, Mastromatteo_2014, Donier_2015, Zarinelli_2015, Said_2017, Lemhadri_2019, Danyliv_2022}.

There is extensive research on the heavy-tailed distribution of financial returns, as documented in numerous studies \cite{Kendall_1953, Mandelbrot_1963, Mandelbrot_1967, Fama_1970, Bollerslev_1987, Andersen_1996, Shephard_1996, Embrechts_1997, Gopikrishnan_1998, Lux_1999, Mittnik_1999, McNeil_2000, Cont_2001, Huisman_2001, Bouchaud_2003, Bradley_2003, Gabaix_2003, Chavez-Demoulin_2012, Taleb_2020}. This topic is also closely related to the substantial literature on stochastic volatility modeling; see, for example, \cite{Andersen_2009, Tsay_2010, Bergomi_2015}.



\section{General Framework}

Let $X \colon \mathcal{X} \rightarrow \mathbb{R}$ be a discrete random variable, where $\mathcal{X} = \{ s_1, \ldots, s_N\}$ is the finite state space. Let $p_i = \mathbb{P}(X=s_i)$ denote the probability of the $i$-th state. The entropy (Shannon entropy) is given by:
\begin{equation}
    \mathcal{S}(X) = - \sum_{i=1}^N p_i \ln p_i.
\end{equation}
We can notice that if $p_i = 1$ for some $i$, which also means that the other states have zero probability, the entropy reaches the minimum value $0$. On the other hand, the maximum value is obtained when each state is equally probable ($p_1 = \ldots = p_N = \frac{1}{N}$). This means that the maximum entropy corresponds to the maximum uncertainty of the system \cite{Guiasu1985}.

Now, assume the system consists of two states: the bid state ($s_b$) and the ask state ($s_a$). Let $p_b = \mathbb{P}(X=s_b)$ and $p_a = \mathbb{P}(X=s_a)$ denote the probabilities of the bid and ask states, respectively. In this case, with only two states, the entropy simplifies to $\mathcal{S}(X) = -p_b\ln{p_b} - p_a\ln{p_a}$, where $p_b + p_a = 1$.

Our aim is to find the probability distribution that maximizes entropy under given constraints. This approach is based on the well-known principle of maximum entropy introduced by Jaynes \cite{Jaynes1957}, which asserts that such a distribution reflects the greatest uncertainty about the system, ensuring that no unjustified assumptions or biases are introduced.
In other words, the maximum entropy distribution represents our best understanding of the system given the available information.

\begin{assumption}
    The state probabilities $p_b = \mathbb{P}(X=s_b)$ and $p_a = \mathbb{P}(X=s_a)$ depend only on the volume imbalance $q$.\footnote{Instead of considering the total order values $Q^bP^b$ and $Q^aP^a$, we focus on the bid and ask imbalances $q^b$ and $q^a$. This is justified for two reasons. First, for highly liquid assets, the bid-ask spread is typically small ($P^b \approx P^a$), so the influence of price differences is minimal. Second, normalizing $Q^b$ and $Q^a$ allows for a standardized comparison across assets with different price levels, ensuring a consistent modeling approach.}
\end{assumption}

In other words, the TOB volume imbalance reflects all information available\footnote{Specifically, information accessible to non-informed traders.} in the market concerning the fundamental value of the asset. For simplicity, we assume a constant spread, leaving the analysis of spread effects on price to future research. Additionally, aggregating information from higher-level order book data could be relevant, particularly for large orders (see, e.g., \cite{Meyer_2019}). Such an analysis, however, is beyond the scope of this study.

We will use standard techniques of Lagrange multipliers to maximize the Shannon entropy under the following constraints:
\begin{equation}
    \begin{cases}
      p_b + p_a = 1\\
      p_b q^b + p_a q^a = \langle q \rangle,
    \end{cases}\,
\end{equation}
where $\langle q \rangle$ is the average of $q$. Let us denote $p_1=p^b$, $p_2 = p^a$ and $q_1 = q^b$, $q_2 = q^a$. The Lagrange function is given by:
\begin{equation}
    L_{\mathcal{S}} = -\left( p_1\ln{p_1} + p_2\ln{p_2} \right) - \lambda \left( p_1 + p_2 - 1 \right) - \beta \left( p_1 q_1 + p_2 q_2 - \langle q \rangle \right).
\end{equation}
Then we have:
\begin{align}
    \frac{\partial L_{\mathcal{S}}}{\partial p_i} = - \ln{p_i} - 1 - \lambda - \beta q_i = 0, \\
    \frac{\partial L_{\mathcal{S}}}{\partial \lambda} = p_1 + p_2 - 1 = 0. \label{eq:prob_sum}
\end{align}
Solving for $p_i$ we get $p_i = e^{-1-\lambda}e^{-\beta q_i}$.
To solve the equation for $\lambda$, we use the constraint~\ref{eq:prob_sum}.
\begin{equation*}
    e^{-1-\lambda} \left( e^{-\beta q_1} + e^{-\beta q_2} \right) = 1,
\end{equation*}
hence $e^{1+\lambda} = \sum_i e^{-\beta q_i}$ where $\sum_i e^{-\beta q_i}$ is the so-called partition function, acting as a normalization factor. Finally, we obtain:
\begin{equation}
    p_i = \frac{e^{-\beta q_i}}{\sum_i e^{-\beta q_i}}.
\end{equation}
In our case, using the previous notation:
\begin{equation}
    p_b = \frac{e^{-\beta q^b}}{e^{-\beta q^b} + e^{-\beta q^a}} \label{eq:prob_b}
\end{equation}
and
\begin{equation}
    p_a = \frac{e^{-\beta q^a}}{e^{-\beta q^b} + e^{-\beta q^a}}. \label{eq:prob_a}
\end{equation}
One might wonder whether $\beta$ can be determined within this framework. In our case, $q^b$ and $q^a = 1-q^b$ are variable, $\langle q \rangle$ is unknown, and $\beta$ remains a free parameter of the model. This situation is analogous to statistical mechanics, where the imbalance can be interpreted as the energy of a given state.\footnote{In kinetic theory of gases, $\beta = \frac{1}{k_B T}$, where $k_B$ is the Boltzmann constant. In this context, $\beta$ can be derived from the known average energy: $\langle E \rangle = \frac{1}{2} k_B T$.}

By maximizing the Shannon entropy subject to certain constraints, we obtain the Boltzmann distribution,\footnote{In some literature, the Boltzmann distribution is referred to as the Boltzmann-Gibbs distribution.} which describes the probability of finding the system in a given state. This distribution can be used to derive the fair price in our model. If we associate the system states $s_b$ and $s_a$ with the bid and ask prices $P^b$ and $P^a$, respectively, then the fair price corresponds to the expected value of the random variable $X \in \{ P^b, P^a\}$,
with respect to the given Boltzmann distribution $\mathbb{P} = \mathbb{P}^{\textit{boltzmann}}$. That is, the fair price can be defined as $\mathbb{E}_{\mathbb{P}^{\textit{boltzmann}}}X$. We refer to this as the Boltzmann price.

Formally, we define the Boltzmann price as follows.

\begin{definition}
Let $\beta \geq 0$,
    \begin{equation}
        P^{\textit{boltzmann}}\left( \beta \right) = \frac{e^{-\beta q^b} P^b + e^{-\beta q^a} P^a}{e^{-\beta q^b} + e^{-\beta q^a}} = \textit{softmax}\left( -\beta q^b, -\beta q^a \right) \cdot \left( P^b, P^a \right),
    \end{equation}
\end{definition}
where $\cdot$ is a dot product, and $\textit{softmax}\left( x_1, \ldots, x_n \right) = \left( \frac{e^{x_1}}{\sum_{i=1}^n e^{x_i}}, \ldots, \frac{e^{x_n}}{\sum_{i=1}^n e^{x_i}}\right)$.\footnote{Equivalently, we can write the Bolztmann price as matrix product $P^{\textit{boltzmann}}\left( \beta \right) = \textit{softmax}\left( -\beta q^b, -\beta q^a \right) \left( P^b, P^a \right)^T$, where $T$ denotes the transpose.}

We observe that, $P^{\textit{boltzmann}}\left( \beta \right) = \textit{softmax}\left( \beta q^a, \beta q^b \right) \cdot \left( P^b, P^a \right)$.

One of the prices, defined below, will be of particular interest to us, as will be justified in the following sections.

\begin{definition}
The equilibrium price is given by:
    \begin{equation}
        P^{\textit{eq}} = P^{\textit{boltzmann}}\left( 1 \right) = \frac{e^{-q^b}P^b + e^{-q^a}P^a}{e^{-q^b} + e^{-q^a}}.
    \end{equation}
\end{definition}

\section{Boltzmann Price Decomposition}

The Boltzmann price can be interpreted as the expected value of the random variable $X \in \{ P^{b}, P^{a}\}$ with a Bernoulli distribution determined by $p_b, p_a$ (see Equations~\eqref{eq:prob_b}, \eqref{eq:prob_a}). In contrast, the micro-price \cite{Stoikov_2018} is defined as the mid-price adjusted by the expected future mid-price changes. Following this perspective, we can decompose the Boltzmann price into the mid-price and an adjustment term. Importantly, this adjustment depends on both the spread and the volume imbalance. As we will demonstrate later, the Boltzmann price can approximate the estimated micro-price in certain cases.

\begin{lemma}[Price Decomposition] \label{prop:price_approx}
    Let $q^b = q =\frac{1}{2} + \theta$, where $\theta \in \left[ -\frac{1}{2}, \frac{1}{2} \right]$. Then the first-order as well as the second-order approximation of the Boltzmann price is given by:
    \begin{equation}
        P^{\textit{boltzmann}}\left( \beta \right) = P^{\textit{mid}} + \frac{\beta \left( P^a - P^b \right)}{2} \theta + \bigO \left( \theta^3 \right)
    \end{equation}
    and
    \begin{equation} \label{eq:mid_w}
        P^{\textit{mid}} + \frac{\beta \left( P^a - P^b \right)}{2} \theta = \left( 1 - \frac{\beta}{2} \right) P^{\textit{mid}} + \frac{\beta}{2} P^{\textit{w}}. 
    \end{equation}
\end{lemma}
\begin{proof}
    We will approximate $P^{\textit{boltzmann}}\left( \beta \right)$ around $q^b=\frac{1}{2}$. By definition,
    \begin{equation*}
        P^{\textit{boltzmann}}\left( \beta \right) = \frac{e^{-\beta q} \cdot P^b + e^{-\beta (1-q)} \cdot P^a}{e^{-\beta q} + e^{-\beta (1-q)}} = \frac{e^{-\frac{\beta}{2}} \left( e^{-\beta \theta} \cdot P^b + e^{\beta \theta} \cdot P^a \right)}{e^{-\frac{\beta}{2}} \left( e^{-\beta \theta} + e^{\beta \theta} \right)}.
    \end{equation*}
    Simplifying the notation, let
    \begin{equation*}
        f(x) = \frac{e^{-\beta x} \cdot b + e^{\beta x} \cdot a}{e^{-\beta x} + e^{\beta x}}.
    \end{equation*}
    To obtain first-order approximation and estimate the error, we need to calculate the first and second derivatives of $f$. We have:
    \begin{equation*}
        f^{\prime}(x) = \frac{2 \beta \left( a-b \right)}{\left( e^{-\beta x} + e^{\beta x}\right)^2}
    \end{equation*}
    and the second derivative of $f$ is equal to
    \begin{equation*}
        f^{\prime\prime}(x) = \frac{4 \beta^2 \left( a-b \right) \left( e^{-\beta x} - e^{\beta x} \right)}{\left( e^{-\beta x} + e^{\beta x}\right)^3}.
    \end{equation*}
    Notice that $f^{\prime\prime}\left( 0 \right) = 0$, so the first- and second-order approximations coincide. Using the Taylor expansion at $x=0$, we obtain,\footnote{The third-order derivative is equal to $f^{(3)}(x) = 8 \beta^3 (a-b) \frac{e^{2\beta x} + e^{-2\beta x} - 4}{\left(e^{-\beta x} + e^{\beta x}\right)^4}$ and we can estimate the approximation error using the Lagrange reminder, which is bounded by $\max_{\xi \in (-\theta, \theta)}\frac{\abs{f^{(3)}(\xi)}}{6} \theta^3$.}
    \begin{equation*}
        f(x) = f(0) + \frac{f^{\prime}(0)}{1!} x + \frac{f^{\prime\prime}(0)}{2!} x^2 + \bigO \left( x^3 \right) = \frac{b+a}{2} + \frac{\beta \left( a - b \right)}{2} x + \bigO \left( x^3 \right).
    \end{equation*}
    To prove Equation~\eqref{eq:mid_w}, we note that
    \begin{align*}
        & P^{\textit{mid}} + \frac{\beta}{2} \left( P^a - P^b \right) \left( q^b - \frac{1}{2} \right) = \frac{P^b + P^a}{2} + \frac{\beta \left[ P^a (q^b - \frac{1}{2}) + P^b (q^a - \frac{1}{2})\right]}{2} = \\
        & \frac{P^a \left( 1 - \frac{\beta}{2} + \beta q^b \right) + P^b \left( 1 - \frac{\beta}{2} + \beta q^a \right)}{2} = \left( 1 - \frac{\beta}{2} \right) \frac{P^b + P^a}{2} + \frac{\beta}{2} \left( P^a q^b + P^b q^a \right).
    \end{align*}
\end{proof}

It is easy to see that
\begin{equation}
    P^{\textit{boltzmann}}\left( 2 \right) \approx P^{\textit{mid}} + \left( P^a - P^b \right) \left( q - \frac{1}{2} \right) = P_{\textit{w}}
\end{equation}
and
\begin{equation}
    P^{\textit{boltzmann}}\left( 0 \right) = P^{\textit{mid}}.
\end{equation}

Furthermore, the smaller the value of $\beta$ and the spread, the better the approximation becomes.

Due to Lemma \ref{prop:price_approx}, we can use $\frac{1}{2} \left( P^{\textit{mid}} + P^{\textit{w}} \right)$ as a good approximation of $P^{\textit{eq}}$. We refer to this as

\begin{definition}
    The quasi-equilibrium price is given by:
    \begin{equation}
        \tilde{P}^{\textit{eq}} = \frac{1}{2} \left( P^{\textit{mid}} + P^{\textit{w}} \right).
    \end{equation}
\end{definition}

\section{Modeling Price Process}

Real-world price processes have been observed to not align fully with Geometric Brownian Motion (GBM) or the Bachelier model \cite{Mandelbrot_1963, Fama_1970, Fama_1995, Cont_2001}. In particular, the distributions of price changes appear to be leptokurtic, meaning they exhibit an excess of observations both near the mean and in the tails of the distribution \cite{Kendall_1953, Fama_1965}.

Mandelbrot \cite{Mandelbrot_1963} suggested the existence of a more general form of the Bachelier model which can account for non-normality. However, the main focus in the research was in non-normal stable distributions with heavier tails, which can account for the empirically observed distributions. For example, \cite{Fama_1965} argued that non-normal stable distributions provide a more accurate description of daily stock returns. Some researchers have also suggested that heavy-tailed empirical distributions may result from mixtures of normal distributions, e.g., \cite{Press_1967, Mandelbrot_1967}.

If we assume that the Boltzmann price reflects the fundamental price in the LOB,
one might also attempt to describe the dynamics of the price process, on a short time scale, as a biased random walk with probabilities given by \eqref{eq:prob_b} and \eqref{eq:prob_a}. For now, we assume that the spread $S_t$ is constant over time, that is, $S_t = S$. In the following, we present the heuristic derivation of price dynamics.

We recall here the definition of the biased random walk.\footnote{\href{https://gordanz.github.io/stochastic-book/}{https://gordanz.github.io/stochastic-book/}}

\begin{definition}
    A stochastic process $\left( Y_n \right)_{n \in \mathbb{N}}$ is a biased random walk with increment $\epsilon > 0$ and parameter $p \in (0,1)$ if
    \begin{enumerate}
        \item $Y_0 = 0$,
        \item the random variables $\xi_1 = Y_1 - Y_0, \xi_2 = Y_2 - Y_1, \ldots$ are independent,
        \item  each $\xi_n$ has two-point distribution,
        \begin{equation}
            \mathbb{P}(\xi_n = \epsilon) = p \quad \text{and} \quad \mathbb{P}(\xi_n = -\epsilon) = 1 - p.
        \end{equation}
    \end{enumerate}
\end{definition}

Let $Y_n$ be a biased random walk with increment $\epsilon$. Calculate its expected value and variance after the $n$ time steps,

\begin{equation}
    \mathbb{E}\left[ Y_n \right] = n \left( p \epsilon - (1-p) \epsilon \right) = n \epsilon \left( 2p - 1 \right)
\end{equation}
and
\begin{equation}
    \text{Var}\left[ Y_n \right] = \mathbb{E}( Y_n^2 ) - \left( \mathbb{E}Y_n \right)^2 = n ( \epsilon^2 - \epsilon^2 \left( 2p-1 \right)^2 ) = 4 n \epsilon^2 p \left( 1-p \right).
\end{equation}

We proceed analogously as in the construction of the Bachelier model; see, e.g., \cite{Joshi_2008}.

Assume now that the fundamental price, denoted here by $P_n$, can change by $\epsilon$ or $-\epsilon$ per unit time with probabilities given by the bid and ask states.


Specifically, the probabilities of price change over one time step are:

\begin{equation}
    \mathbb{P}\left( P_{n+1} - P_{n} = \epsilon \right) = \frac{e^{\beta \theta_{n}}}{e^{\beta \theta_n} + e^{-\beta \theta_n}},
\end{equation}
\begin{equation}
    \mathbb{P}\left( P_{n+1} - P_{n} = -\epsilon \right) = \frac{e^{-\beta \theta_{n}}}{e^{\beta \theta_n} + e^{-\beta \theta_n}}.
\end{equation}

We observe that the probabilities of price changes may vary with each time step. Let us compute the expected value and the variance of the next price change. We have:

\begin{equation}
    \mu_n =\mathbb{E}\left[ P_{n+1} - P_{n} \right] = \frac{e^{\beta \theta_{n}}}{e^{\beta \theta_n} + e^{-\beta \theta_n}} \cdot\epsilon - \frac{e^{-\beta \theta_{n}}}{e^{\beta \theta_n} + e^{-\beta \theta_n}} \epsilon = \epsilon \cdot \frac{e^{\beta \theta_n} - e^{-\beta \theta_n}}{e^{\beta \theta_n} + e^{-\beta \theta_n}} = \epsilon \tanh{\left( \beta \theta_n \right)}
\end{equation}
and the variance is equal to:
\begin{equation}
    \sigma_n^2 = \Var\left[ P_{n+1} - P_{n} \right] = 4 \epsilon^2 \frac{e^{\beta \theta_n} \cdot e^{-\beta \theta_n}}{\left( e^{\beta \theta_n} + e^{-\beta \theta_n} \right)^2} = \frac{\epsilon^2}{\cosh^2{\left( \beta \theta_n \right)}}.
\end{equation}


Now, consider the time horizon $[0,T]$ and divide it into $N$ equal intervals $\Delta t = \frac{T}{N}$. To keep $\mu_n$ and $\sigma_n$ fixed as we vary $N$, we scale them to $\mu_n \Delta t$ and $\sigma_n \sqrt{\Delta t}$ for each time interval $\Delta t$. This ensures that the mean and standard deviation of the process at each step are appropriately adjusted.\footnote{Since mean and variance add for independent variables, the total mean is $\sum_{n=1}^N \mu_n \Delta t$, and the total variance is $\sum_{n=1}^N \sigma_{n}^{2} \Delta t$. If $\mu_n = \mu$ and $\sigma_n = \sigma$ are constant for all $n$, the total mean and variance simplify to $\mu t$ and $\sigma^{2} t$, respectively, where $t = N \Delta t$, consistent with the standard case.}

Formally, one can decompose the price process into a drift component and an unbiased term. At each time step $\Delta t$, the price can jump up or down with equal probability, defined as:
\begin{equation}
    \mathbb{P}\left( P_{i+\Delta t} - P_{i} = \mu_i \Delta t + \sigma_i \sqrt{\Delta t} \right) = \frac{1}{2},
\end{equation}
\begin{equation}
    \mathbb{P}\left( P_{i+\Delta t} - P_{i} = \mu_i \Delta t - \sigma_i \sqrt{\Delta t} \right) = \frac{1}{2}.
\end{equation}
We find that:
\begin{equation}
    \mathbb{E}\left[ P_{i+\Delta t} - P_{i} \right] = \mu_i \Delta t
\end{equation}
and
\begin{equation}
    \Var\left[ P_{i+\Delta t} - P_{i} \right] = \sigma_i^2 \Delta t.
\end{equation}
Given a discrete-time price process with time-varying drift and volatility, we have
\begin{equation}
    P_{n+1} = P_n + \mu_n \Delta t + \sigma_n \sqrt{\Delta t } Z_n,
\end{equation}
where $Z_n$ has the Bernoulli distribution with $\mathbb{E}\left( Z_n \right) = 0$ and $\Var\left( Z_n \right) = 1$. The fundamental price after $n$ time steps of size $\Delta t$ is given by:
\begin{equation}
    P_N = P_0 + \sum_{i=0}^{N-1} \mu_i \Delta t + \sum_{i=0}^{N-1} \sigma_i \sqrt{\Delta t} Z_i.
\end{equation}
Hence,
\begin{equation}
    P_N = P_0 + \sum_{i=0}^{N-1} \epsilon\tanh{\left( \beta \theta_i \right)} \Delta t + \sum_{i=0}^{N-1} \frac{\epsilon}{\cosh{\left( \beta \theta_i \right)}} \sqrt{\Delta t} Z_i
\end{equation}
and
\begin{equation}
    \mathbb{E}\left[ P_N \right] = S_0 + \epsilon \sum_{i=0}^{N-1} \tanh{\left( \beta \theta_i \right)} \Delta t,
\end{equation}
\begin{equation}
    \Var\left[ P_N \right] = \epsilon^2 \sum_{i=0}^{N-1} \frac{1}{\cosh^2{\left( \beta \theta_i \right)}} \Delta t.
\end{equation}
By the \textit{Central Limit Theorem} when $N$ goes to infinity and $\Delta t \to 0$, the distribution of $\sum_{i=0}^{N-1} \sqrt{\Delta t} Z_i$ will tend to a Gaussian random variable, with mean $0$ and variance $1$.

Having $T = N \Delta t$ and assuming $N \to \infty$ and $\Delta t \to 0$, we get the following:

\begin{equation}
    \sum_{i=0}^{N-1} \mu_i \Delta t \rightarrow \int_0^T \mu_s ds
\end{equation}
and
\begin{equation}
    \sum_{i=0}^{N-1} \sigma_i \sqrt{\Delta t} Z_i \rightarrow \int_0^T \sigma_s dW_s.
\end{equation}
It follows that:
\begin{equation}
    P_T = P_0 + \int_0^T \mu_s ds + \int_0^T \sigma_s dW_s.
\end{equation}
Hence, the stochastic differential equation of the price process can be written as:
\begin{equation}
    dP_t = \epsilon\tanh{\left( \beta \theta_t \right)} dt + \frac{\epsilon}{\cosh{\left( \beta \theta_t \right)}} dW_t.
\end{equation}
Now, observe that the second-order approximation of $f(x) = \epsilon \tanh{\left( \beta x \right)}$ around $x=0$ is given by:
\begin{equation}
    \epsilon \tanh{\left( \beta \theta \right)} = \beta \epsilon \theta + \bigO \left( \theta^3 \right).
\end{equation}
We find a correspondence with Lemma~\ref{prop:price_approx} (Price Decomposition). A natural candidate for $\epsilon$ is half the spread, $\frac{S}{2}$. For such $\epsilon$, the drift term represents the adjustment to the mid-price given by the Boltzmann price.\footnote{One can also argue that, when the TOB is balanced, then both bid and ask probabilities are equal, and the Boltzmann price coincides with the mid price. Then it is natural for the price to jump by half the spread with some probabilities. So, the average move should be proportional to half the spread.}

This is analogous to Roll's model \cite{Roll_1984}, in which the observed market price $P_t^{\textit{market}}$ is equal to
\begin{equation}
    P_{t}^{\textit{market}} = P_{t}^{\ast} + \zeta_t \frac{S}{2},
\end{equation}
where $P_t^{\ast}$ denotes the fundamental asset value, and $\zeta_t \in \{ -1, 1\}$ is a sequence of i.i.d. random variables with the Rademacher distribution $\mathbb{P}(\zeta_t = 1) = \mathbb{P}(\zeta_t = -1) = \frac{1}{2}$.\footnote{See \cite{Tsay_2010} for further details.}

In our case, assuming a constant spread, we obtain:
\begin{equation} \label{eq:price_approx}
    P_t^{\textit{boltzmann}} \approx P_t^{\textit{mid}} + \beta \theta_t \frac{S}{2}.
\end{equation}
In general, the fundamental price can fluctuate with varying frequency and magnitude. In continuous-time modeling, we aim to define a general-form equation as follows:
\begin{equation} \label{eq:price_dynamics}
    dP_t = \sigma \left( \tanh{\left( \beta \theta_t \right)} dt + \frac{dW_t}{\cosh{\left( \beta \theta_t \right)}} \right),
\end{equation}
and we might also add an additional drift term $\mu dt$ and estimate parameters $\mu$ and $\sigma$ from data.\footnote{Over short time periods, the drift term $\mu dt$ is not justified within our framework. However, if we consider longer time horizons, it may be appropriate to include factors such as the time value of money, which makes the additional drift term natural.} Here, the spread is assumed constant.
Without the assumption of constant spread, the equation might be defined as follows:
\begin{equation}\label{eq:price_dynamics_varying_spread}
    dP_t = \eta S_t \left(\tanh{\left( \beta \theta_t \right)} dt + \frac{dW_t}{\cosh{\left( \beta \theta_t \right)}}\right).
\end{equation}
It is important to note that, in practice, the drift term $dt$ and the volatility term $dW_t$ scale differently---namely with $\Delta t$ and $\sqrt{\Delta t}$, respectively---resulting in different constants when a specific time step $\Delta t$ is considered.

It is clear that as $\beta \to 0$, Equation~\eqref{eq:price_dynamics} reduces to the standard Bachelier model without a drift term, that is, $dP_t = \sigma dW_t$.

In our model, the volatility $\sigma$ should be proportional to half the spread $\frac{S}{2}$, i.e., $\sigma \sim \frac{S}{2}$. This aligns with the existing literature, where, for example, it is observed that the spread is proportional to the volatility per trade. \cite{Wyart_2008} Generally, higher volatility implies greater uncertainty and risk for market participants. It is noted that to compensate for that risk the spread is widened.

Since we consider a short time period, the difference between the Bachelier model and the GBM is negligible. However, an equivalent GBM-like equation could be formulated as follows:\footnote{We leave this model for future work. Using similar reasoning, i.e., modeling the price as an exponential biased random walk, we may attempt to derive this equation.}
\begin{equation} \label{eq:GBM_equivalent}
    dP_t = \sigma P_t \left( \tanh{\left( \beta \theta_t \right)} dt + \frac{dW_t}{\cosh{\left( \beta \theta_t \right)}} \right),
\end{equation}
which again approximates the standard Black-Scholes equation without drift as $\beta \to 0$. We may also include an additional drift term, $\mu P_t dt$. However, Equation~\eqref{eq:GBM_equivalent} fundamentally contains only two parameters, $\sigma$ and $\beta$, as in the standard GBM model.

We leave a formal study of the proposed equations to future research. For now, focusing solely on Equation~\eqref{eq:price_dynamics}, one may assume that $\theta_t$ is continuous and adapted to the filtration generated by $W_s(\cdot)$, $s \leq t$ and $P_0$. However, the assumption of continuity for $\theta_t$ may not be realistic, as it tends to fluctuate significantly in practice. A more natural assumption is to take $\theta_t$ as a predictable process. In this context, predictability reflects the idea that the volume imbalance prior to time $t$ influences the price at time $t$. This observation is consistent with findings in the literature. For example, \cite{Cont_2014} show that, over short time intervals, price changes are primarily driven by the order flow imbalance at the best bid and ask. Similarly, \cite{Gould_2016} demonstrate that queue imbalance carries significant predictive power for the direction of the next price movement. Furthermore, \cite{Toth_2015} find that order flow in equity markets exhibits persistence---buy orders tend to be followed by more buy orders, and likewise for sell orders.
These findings support the view that the volume imbalance at time $t$ is highly dependent on the state of the TOB shortly before time $t$, justifying its modeling as a predictive process.\footnote{Both functions appearing in Equation~\eqref{eq:price_dynamics} are bounded, i.e., $\abs{\tanh{x}} \leq 1$ and $\abs{\frac{1}{\cosh{x}}} \leq 1$, satisfying the conditions of the existence and uniqueness theorem for stochastic differential equations \cite{Oksendal_2010}. It is natural to consider the spread to be bounded. To model potential discontinuities, one may assume, for example, that $S_t$, $\theta_t$ are c{\`a}dl{\`a}g. However, this involves accounting for the left limits.}

The main analysis is structured into two components: simulations and historical data analysis, which are presented in the following sections.

\section{Simulation Analysis}

This section presents a comparison between the Bachelier model and the proposed dynamics of the price process described in Equations~\eqref{eq:price_dynamics} and \eqref{eq:price_dynamics_varying_spread}. The primary focus is on examining heavy-tailed distributions in price changes and market impact. Given the availability of historical data, we estimate both the spread and the imbalance. The imbalance is modeled using a Beta distribution, while the spread is modeled using a Gamma distribution, in line with the approach of \cite{Lo_2002}.\footnote{We use the parametrization of the Gamma distribution with a shape parameter and a scale parameter.}

Throughout this paper, we refer to excess kurtosis (also known as Fisher kurtosis), defined as $\textit{Kurt}(X)-3$, where
\begin{equation}
    \textit{Kurt}(X) = \mathbb{E}\left[ \frac{X - \mu}{\sigma} \right]
\end{equation}
and $\mu$ and $\sigma$ are the mean and the standard deviation of the random variable $X$, respectively.\footnote{Kurtosis is computed using the \texttt{kurtosis(Fisher=True)} function from the \texttt{scipy.stats} library.}

\subsection{Simulation with Varying Spread}

In this example, we choose symmetric imbalance centered and concentrated around $\frac{1}{2}$, specifically $\textit{Beta}(4.5,4.5)$.
For the spread, we sample from the distribution $\Gamma(1,1)$. We assume $\beta = 1$ and set the tick size to 0.01. To ensure that the spread is a multiple of the tick size, we take the ceiling of each sampled value and multiply it by 0.01. In Figure~\ref{fig:sim_imb_spread_beta1}, we display the sampled spread and imbalance from a representative simulation.

\begin{figure}[htbp]
    \centering
    \subfloat[\centering Imbalance $q$ sampled from $\textit{Beta}(4.5,4.5)$.]{{\includegraphics[scale=0.4]{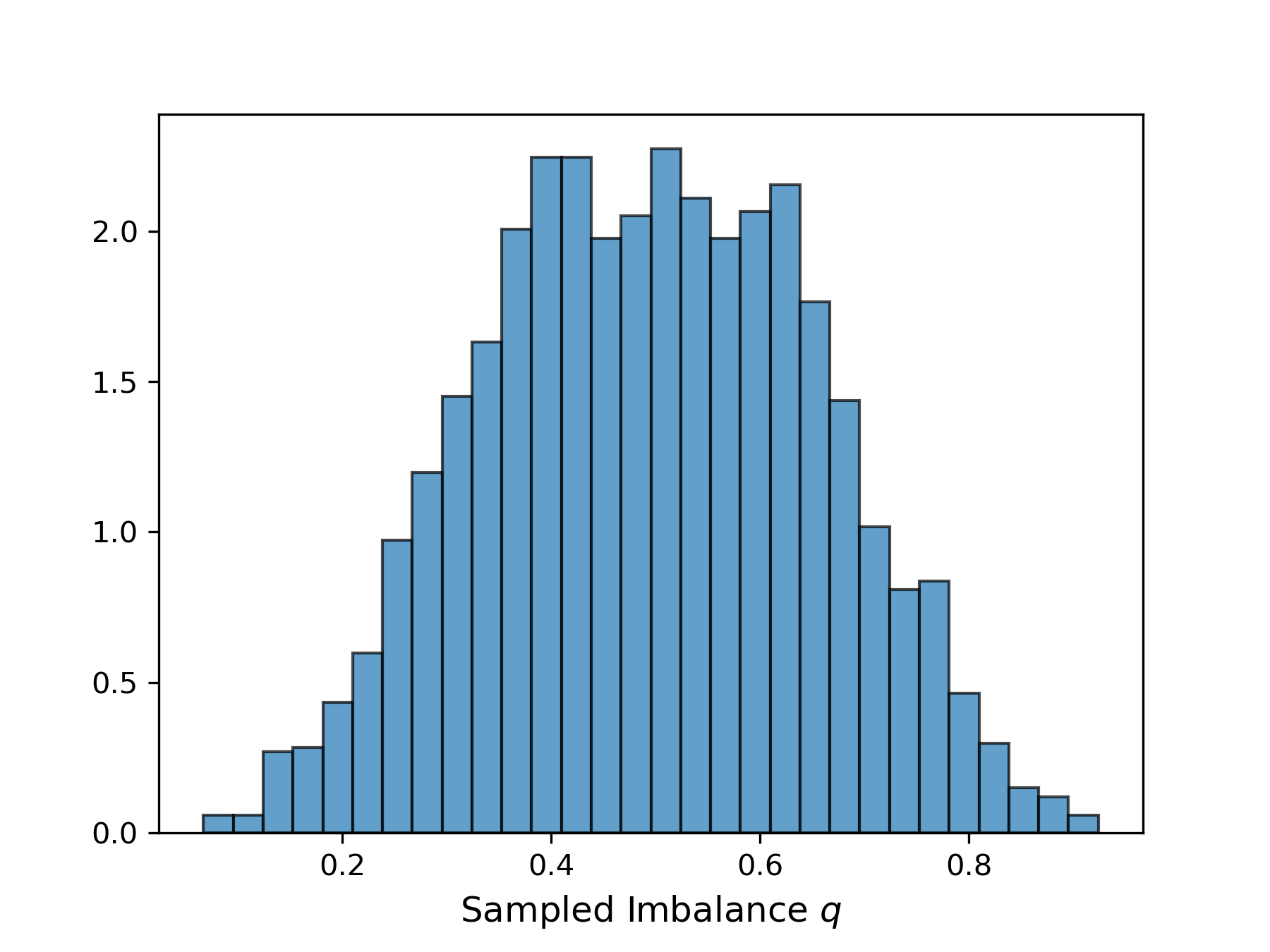} }}%
    \qquad
    \subfloat[\centering Spread sampled from $\Gamma(1,1)$.]{{\includegraphics[scale=0.4]{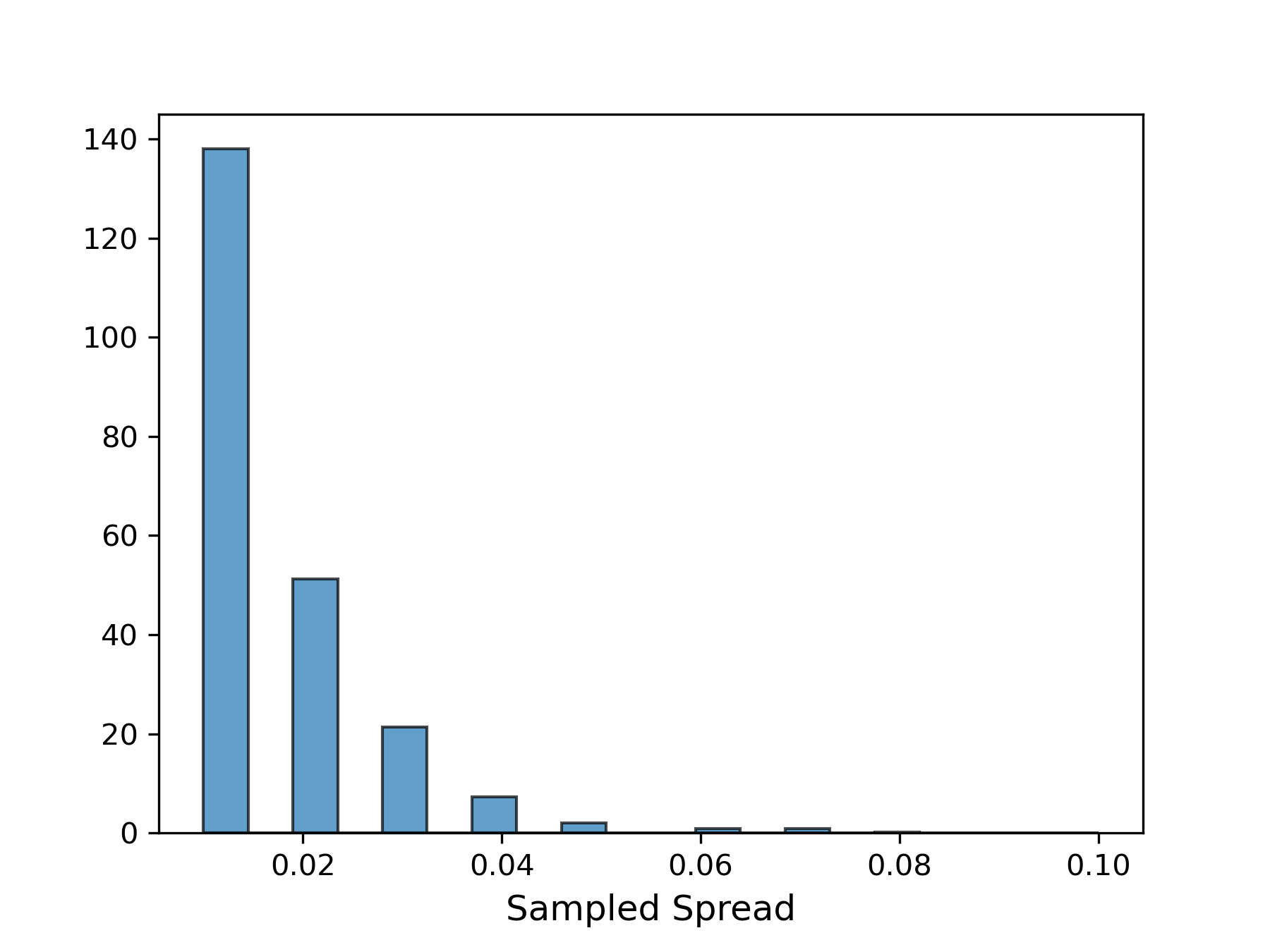} }}%
    \caption{Sampled imbalance and spread used in a representative price simulation with $\beta=1$.}
    \label{fig:sim_imb_spread_beta1}
\end{figure}

For this simulation we take 2,340 time steps with time increment $\Delta t = \frac{1}{N}$. This, e.g., can be interpreted as 6.5 hours of trading (23,400s) with the time step 10 seconds.
\Cref{tab:sim_varying_spread_kurtosis_1000_summary} presents the kurtosis statistics of price changes across 1000 simulation runs. In Equation~\eqref{eq:price_dynamics_varying_spread} we use the scaling parameter $\eta$ which was selected to provide similar means and standard deviations of price changes for both models; see \Cref{tab:sim_varying_spread_mean_and_std}.

\begin{table}[htb]
    \begin{minipage}{.48\linewidth}
      \centering
        \begin{tabular}{l c}
\hline
Statistic & Kurtosis \\
\hline
Mean      & 7.29 \\
Std       & 3.16 \\
Min       & 3.18 \\
Max       & 41.8 \\
\hline
\\
\end{tabular}
        \caption{Summary statistics of kurtosis values, rounded
        to two decimal places, from 1,000 simulations
        with $\beta=1$, where imbalance is sampled from
        a $\textit{Beta}(4.5,4.5)$ distribution and the spread
        is sampled from a $\Gamma(1,1)$ distribution.}
                \label{tab:sim_varying_spread_kurtosis_1000_summary}
    \end{minipage}
    \hspace{0.04\linewidth}
    \begin{minipage}{.48\linewidth}
      \centering
        \begin{tabular}{l cc}
        \\
        \\
\hline
Statistic & Bachelier $\Delta P$ & Model $\Delta P$\\
\hline
Mean    & \num{-7.99e-7} & \num{-7.87e-7}  \\
Std     & \num{1.034e-3} & \num{1.038e-3}  \\
\hline
\\
\end{tabular}
        \caption{Mean and standard deviation of price changes, averaged over 1,000 simulations, from our model with $\beta = 1$, compared to the Bachelier model.    
        }
        \label{tab:sim_varying_spread_mean_and_std}
    \end{minipage}
\end{table}



Figure~\ref{fig:price_sim_with_spread_beta1} shows histograms and Kernel Density Estimations (KDEs) of price changes.

\begin{figure}[htbp]
    \centering
    \subfloat[\centering Histogram of price changes]{{\includegraphics[scale=0.4]{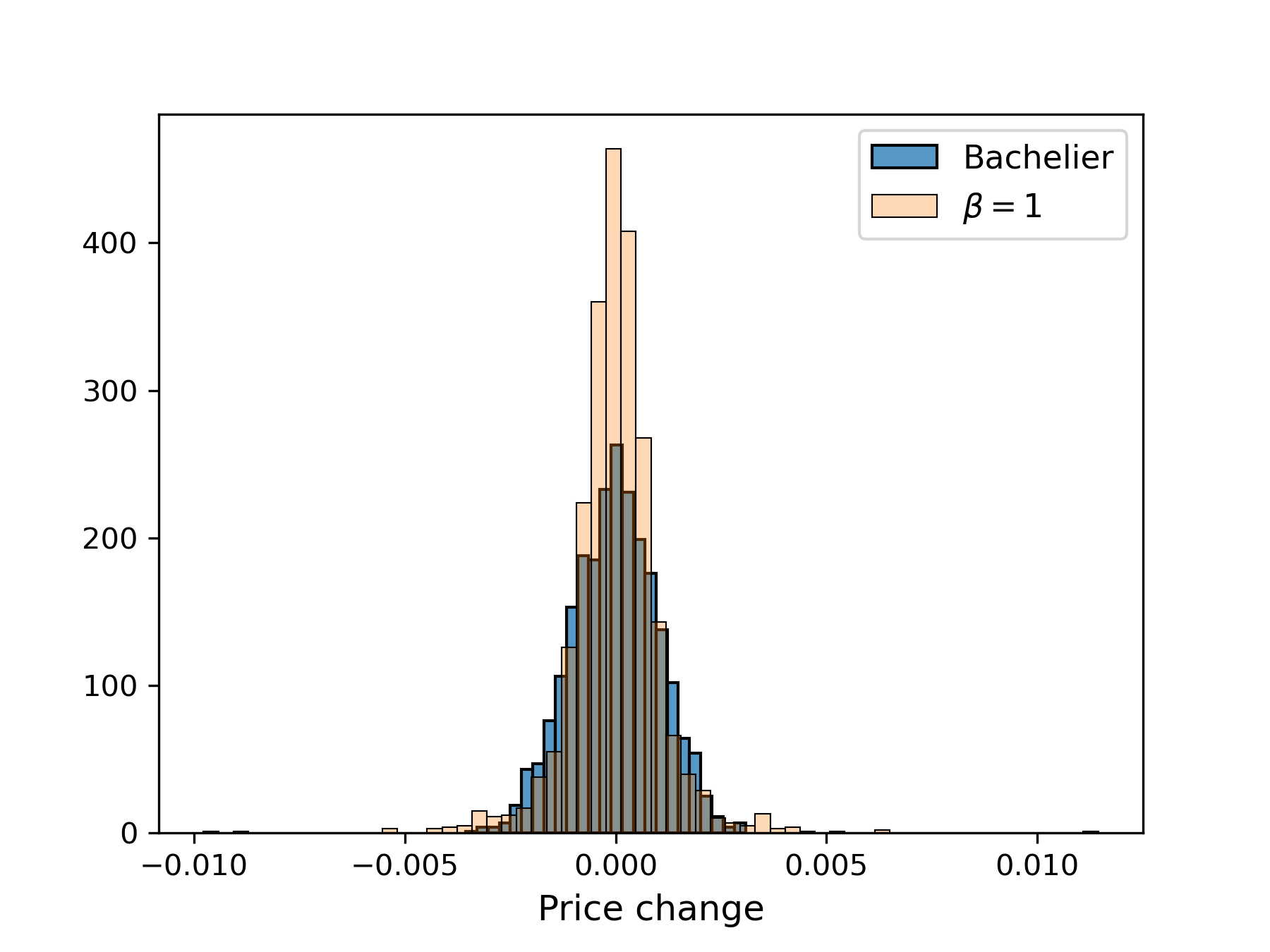} }}%
    \qquad
    \subfloat[\centering KDE of price changes]{{\includegraphics[scale=0.4]{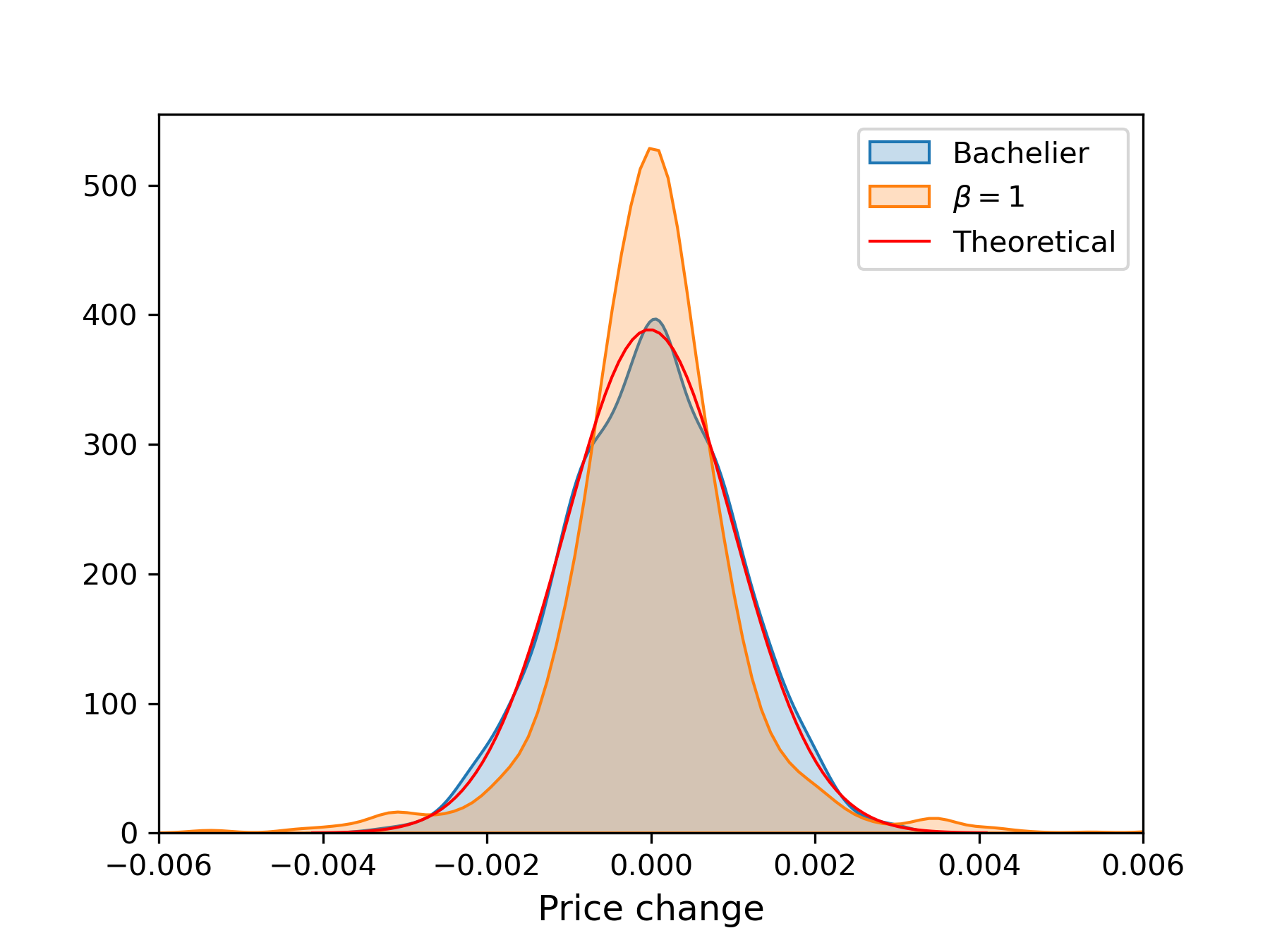} }}
    \caption{Price changes from the Bachelier model without drift, $dP_t = \sigma dW_t$, compared to those generated by Equation~\eqref{eq:price_dynamics_varying_spread},
    with $P_{0} = 10$, $q_t \sim \textit{Beta}\left( 0.5,0.5 \right)$, $\sigma = 0.05$, $\eta = 2.75$, and $\theta_t = q_t -0.5$. Wiener process increments $dW_t$ are shared between both models. In this example, the simulated price changes exhibit a kurtosis of 12.95. In contrast, the standard Bachelier model yields price changes with a kurtosis close to 0.}
    \label{fig:price_sim_with_spread_beta1}
\end{figure}

In this particular example, the simulated price changes yield a kurtosis of 12.95. In contrast, the standard Bachelier model produces price changes with a kurtosis -0.09.
Specifically, the standard deviation of both simulated price changes is approximately 0.001.

As suggested by stochastic volatility models, changing volatility (here modeled by spread) is resulting in heavy tails; see, for example, \cite{Andersen_2009, Tsay_2010, Bergomi_2015}.

\subsection{Simulation with Constant Spread}\label{ssec:sim_const_spread}

Modeling the price process under a constant spread assumption presents additional challenges. We observe that the tails of price returns distributions become noticeably heavier as the parameter $\beta$ increases. However, a more critical factor appears to be the volume imbalance, particularly when the TOB displays a significant volume concentration on one side or follows a $U$-shaped distribution.

\begin{figure}[htbp]
    \centering
    \includegraphics[scale=0.4]{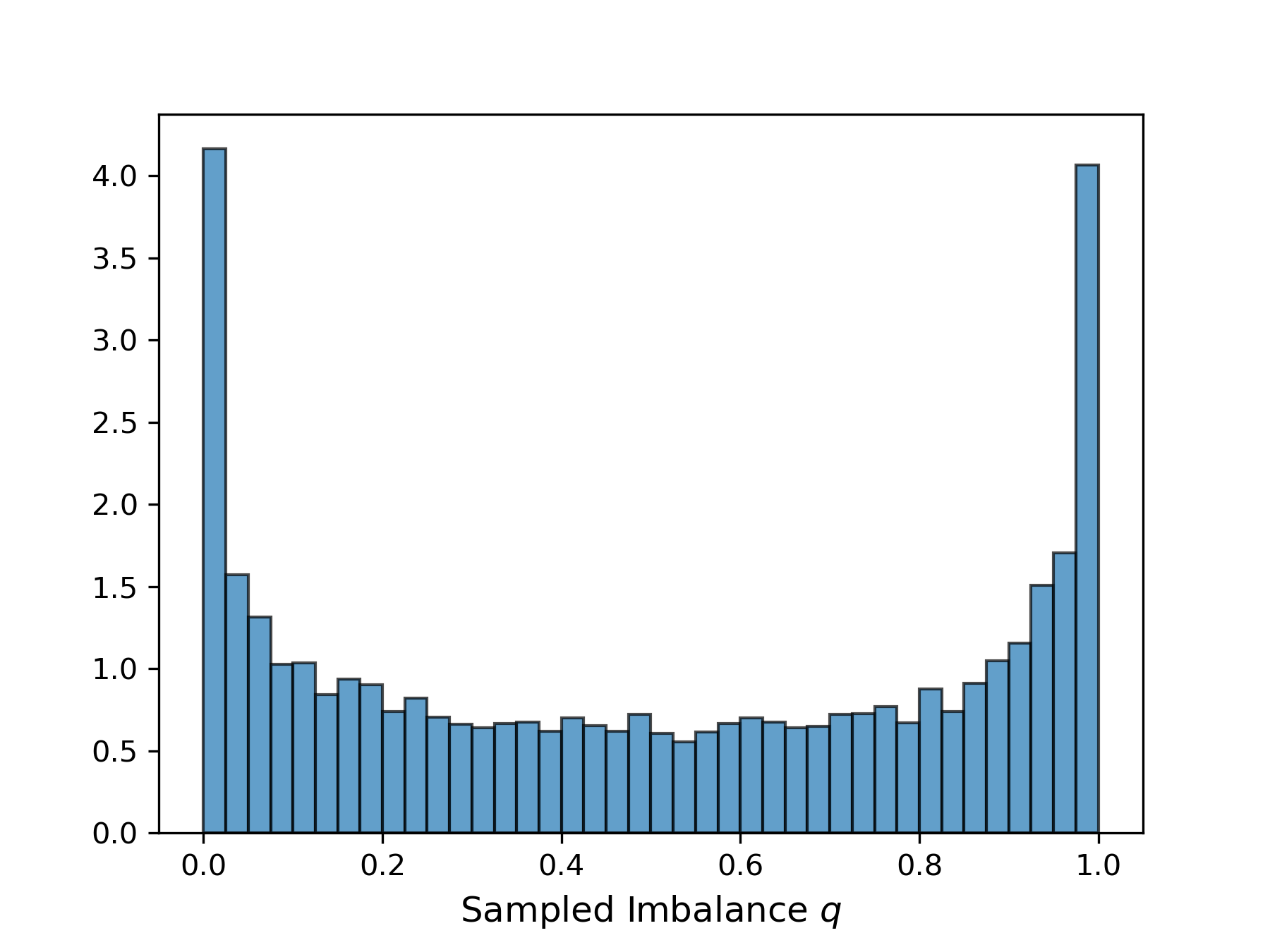}
    \caption{Histogram of imbalance sampled from a $\textit{Beta}(0.5,0.5)$ distribution.}
    \label{fig:hist_imb_sampled_beta0_5_0_5}
\end{figure}

In this example, we model the imbalance using a $U$-shaped distribution, as illustrated in Figure~\ref{fig:hist_imb_sampled_beta0_5_0_5}. Such modeling may be justified by observations in, for example, \cite{Lehalle_2017}, where the order book imbalance---particularly just before a trade---often appears to follow a $U$-shaped distribution.

Figure~\ref{fig:price_sim_const_spread_beta5} shows a representative price simulation for $\beta = 5$. We compare a standard Bachelier model without drift, i.e., $dP_t = \sigma dW_t$ with the price generated from Equation~\eqref{eq:price_dynamics}, i.e., $dP_t = \tilde{\sigma} ( \tanh{(\beta \theta_t)} dt + \frac{dW_t}{\cosh{(\beta \theta_t)}} )$,
where $P_{0} = 10$, $q_t \sim \textit{Beta}\left(0.5,0.5\right)$ and $\sigma = 0.25$, $\tilde{\sigma} = 2\sigma=0.5$.
In each simulation, the Wiener process increments $dW_t$ are shared between both models.
The simulation is run with 8,000 time steps over a time horizon of $T=1$ (i.e., $\Delta t = \frac{1}{8000}$), corresponding to approximately 1 month of trading at 1-minute intervals.
In this example, the simulated price changes generated by Equation~\eqref{eq:price_dynamics} yield a kurtosis of 4.67. In contrast, the standard Bachelier model produces price changes with a kurtosis of approximately 0.16. To obtain comparable means and standard deviations of price changes, we scale $\sigma$ by a factor of two, i.e., $\tilde{\sigma} = 2\sigma$.
Specifically, the average standard deviations of price changes across 1,000 simulations are nearly identical: 0.0028 for the Bachelier model and 0.0029 for the model given by Equation~\eqref{eq:price_dynamics}.

\Cref{tab:kurtosis_1000_summary} presents the kurtosis statistics of price changes across 1000 simulation runs for various imbalance distributions. The results indicate that both significantly skewed imbalances and $U$-shaped distributions contribute to increased kurtosis in the price returns.

\begin{table}[htb]
\centering
\begin{tabular}{l ccc}
\hline
Kurtosis & $\textit{Beta}(0.5,0.5)$ & $\textit{Beta}(2,2)$ & $\textit{Beta}(8,2)$ \\
 & $\beta=5$ & $\beta=7.5$ & $\beta=7.5$ \\
\hline
Mean      & 3.83 & 2.5 & 8.75 \\
Std       & 0.29 & 0.19 & 0.8 \\
Min       & 2.93 & 1.9 & 6.64 \\
Max       & 4.85 & 3.22 & 11.98 \\
\hline
\\
\end{tabular}
\caption{Summary statistics of kurtosis values from 1,000 simulations where imbalance is sampled from $\textit{Beta}$ distributions with different parameters.}
\label{tab:kurtosis_1000_summary}
\end{table}

The resulting price changes generated by Equation~\eqref{eq:price_dynamics} exhibit heavy tails, likely driven by the combination of a $U$-shaped $\textit{Beta}$ imbalance distribution and the high value of the parameter $\beta$.


\begin{figure}[htbp]
    \centering
    \subfloat[\centering Price paths\label{fig:price_sim_const_spread_beta5_a}]{{\includegraphics[scale=0.42]{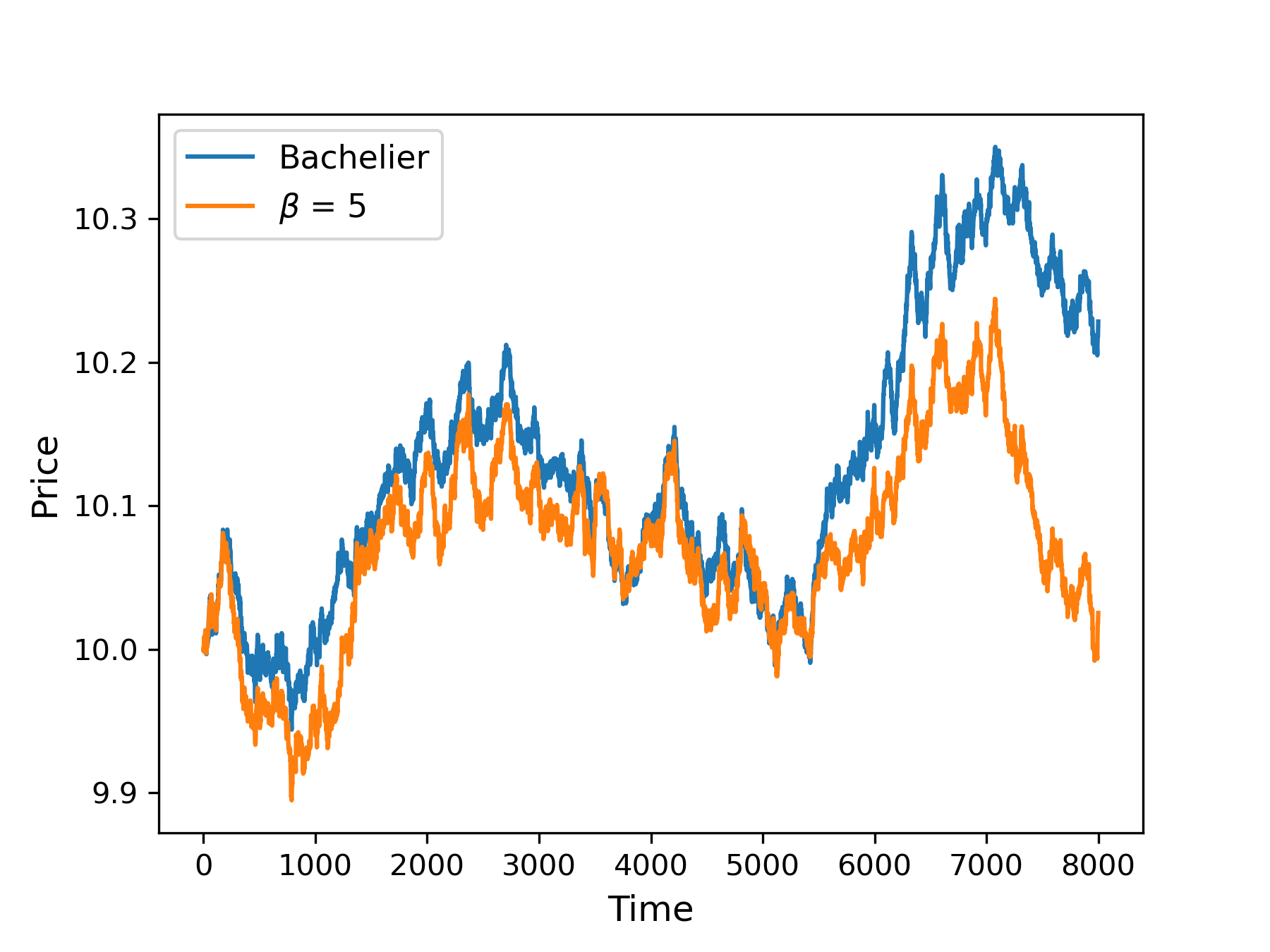} }}
    \qquad
    \subfloat[\centering Histogram of price changes\label{fig:price_sim_const_spread_beta5_b}]{{\includegraphics[scale=0.42]{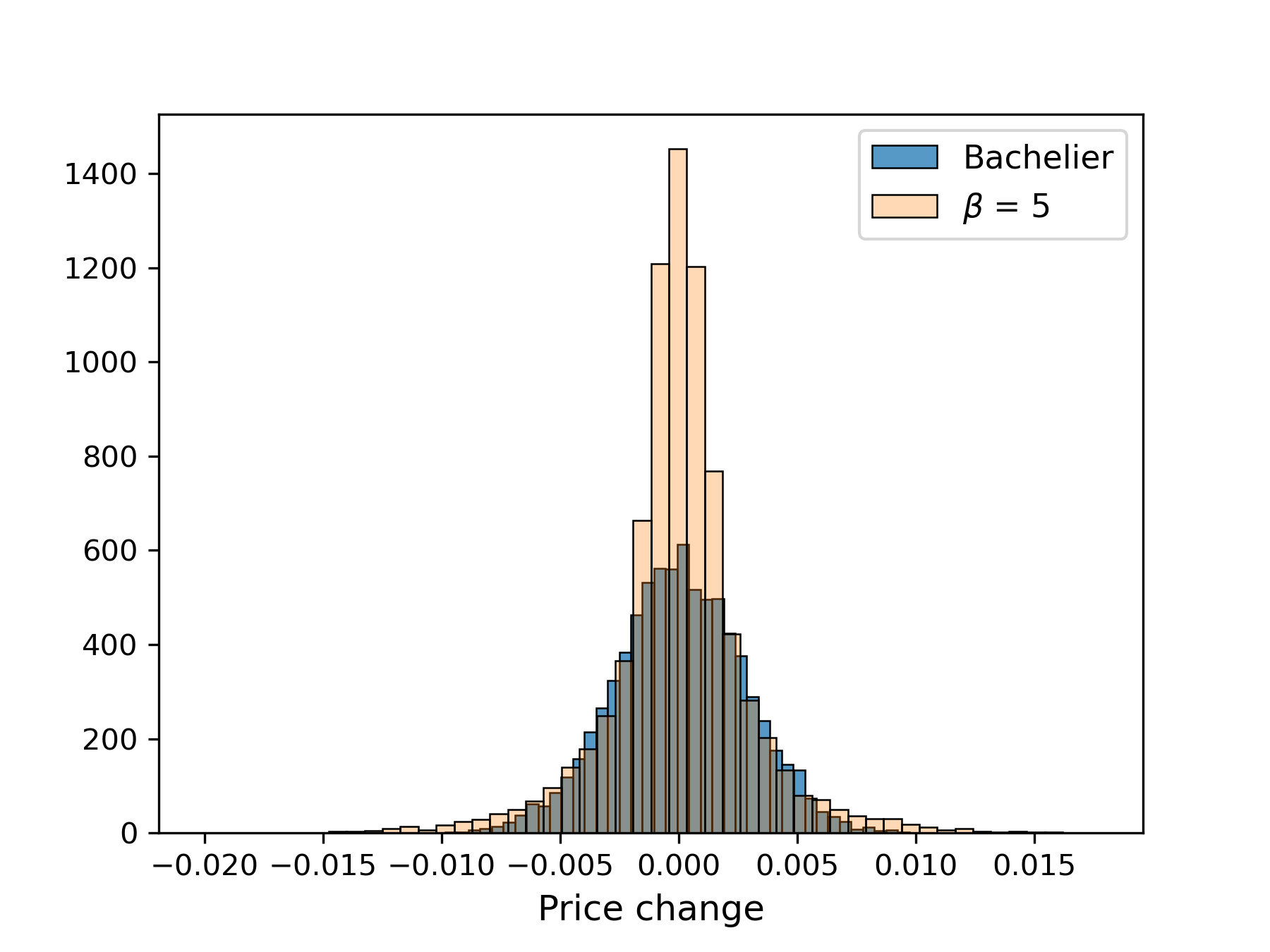} }}
    \caption{Price simulation: We compare the Bachelier model without drift, given by, $dP_t = \sigma dW_t$ to the price generated by Equation~\eqref{eq:price_dynamics}
    , i.e., $dP_t = \tilde{\sigma} ( \tanh{(\beta \theta_t)} dt + \frac{dW_t}{\cosh{(\beta \theta_t)}} )$, where the initial price is $P_{0} = 10$, $\beta = 5$ the imbalance $q_t \sim \textit{Beta}\left( 0.5, 0.5 \right)$, and $\sigma = 0.25$, $\tilde{\sigma} = 0.5$.
    Wiener process increments $dW_t$ are shared between both models.
    In this example, the simulated price changes generated by Equation~\eqref{eq:price_dynamics} exhibit a kurtosis of 4.67. In contrast, the standard Bachelier model produces price changes with a kurtosis of approximately 0.16.}
    \label{fig:price_sim_const_spread_beta5}
\end{figure}

\subsection{Market Impact Simulation} \label{ssec:market_impact_sim}

A common approach in market microstructure is to model the mid-price process as a random walk incorporating a permanent market impact component, as described in \cite{Gueant_2016}. The dynamics can be expressed as:

\begin{equation}
    dP_t^{\textit{mid}} = \sigma dW_t + kv_t dt,
\end{equation}

where $v_t$ is the trading velocity, $\sigma > 0$ represents the volatility of the stock and $k \geq 0$ is the magnitude of the permanent market impact. When $k=0$, indicating no permanent market impact, $P_t^{\textit{mid}}$ follows a purely random walk. In Figure~\ref{fig:price_sim_drift} we can observe that the drift term can arise from bid and ask volume imbalances.

\begin{figure}[!htbp]
    \centering
    \includegraphics[scale=0.4]{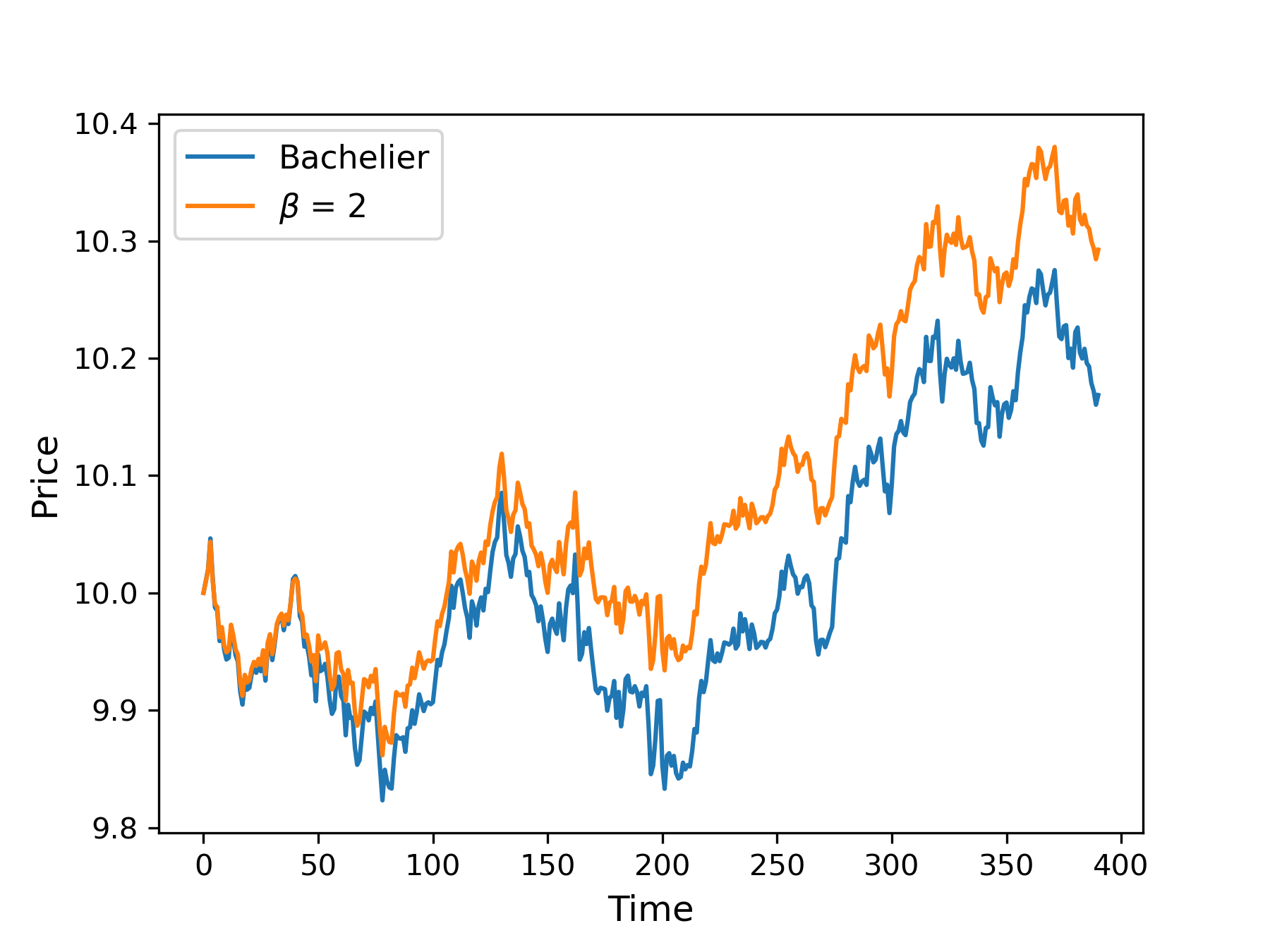}
    \caption{Price simulation: We compare the Bachelier model with drift, given by, $dP_t = \mu dt + \sigma dW_t$ to the price generated by Equation~\eqref{eq:price_dynamics} with $\beta=2$,
    where $P_{0} = 10$, $\mu = 0.1$, the imbalance $q_t \sim \textit{Beta}\left( 6.733, 3.267 \right)$ and $\sigma = 0.3$. In each simulation, the Wiener process increments $dW_t$ are shared between both models. The time horizon is set to 1, with a simulation time step of$\frac{1}{N}$, where $N=390$, corresponding to 6.5 trading hours at 1-minute intervals.
    The plot indicates that in our model the drift term can arise from volume imbalance.}
    \label{fig:price_sim_drift}
\end{figure}

One can easily estimate parameters of a Beta distribution that yield approximately the same drift $\mu$ as in the Bachelier model with drift, i.e., $dP_t = \mu dt + \sigma dW_t$. Because each time step $\theta_{i}$ is sampled from the same distribution, we have:
\begin{equation}
    \sigma \mathbb{E}\tanh{\left( \beta \theta_{i} \right)} = \mu.
\end{equation}
The mean of a $\textit{Beta}(a,b)$ distribution is given by $\frac{a}{a+b}$; hence, we obtain
\begin{equation}
    \hat{\theta}_{i} = \frac{1}{\beta}\tanh^{-1}{\left( \frac{\mu}{\sigma} \right)}.
\end{equation}
For example, let $\beta = 2$, $\mu =0.1$, and $\sigma = 0.3$. This yields $\hat{\theta}_{i} = \frac{1}{2} \tanh^{-1}{\left( \frac{1}{3} \right)} \approx 0.173$, for which a $\textit{Beta}(6.733,3.267)$ distribution serves as a good approximation.

In this example, to simulate the drift $\mu$ in the model given by Equation~\eqref{eq:price_dynamics}, we set $\beta=2$, $P_{0} = 10$, $\mu = 0.1$, $\sigma = 0.3$, and the imbalance $q_t \sim \textit{Beta}\left( 6.733, 3.267 \right)$. In each simulation, the Wiener process increments $dW_t$ are shared between both models. The time horizon is set to 1, with a simulation time step of $\frac{1}{N}$, where $N=390$, corresponding to 6.5 trading hours at 1-minute intervals. \Cref{tab:market_impact_price_statistics} presents summary statistics averaged across 1,000 simulation runs. We observe average price changes of approximately 0.1 that are similar for both models.

\begin{table}[htb]
\centering
\begin{tabular}{l cc}
\hline
Statistic & Bachelier & $\textit{Beta}\left( 6.733,3.267 \right)$ \\
\hline
Mean      & 10.1078 & 10.1006 \\
Std       & 0.3 & 0.28 \\
Min       & 9.17 & 9.24 \\
Max       & 11.05 & 10.93 \\
\hline
\\
\end{tabular}
\caption{Summary statistics for the price value after 8,000 times steps, averaged across 1,000 simulations. The mean is reported to four decimal places, while all other statistics are rounded to two decimal places.
}
\label{tab:market_impact_price_statistics}
\end{table}


\subsection{Simulations with Symmetric Imbalance Concentrated at 0.5}\label{ssec:beta_and_symm_imb_sim}

In this example, we compare the behavior of the Bachelier model without drift, given by, $dP_t = \sigma dW_t$, and the price process generated by Equation~\eqref{eq:price_dynamics}, i.e., $dP_t = \tilde{\sigma} ( \tanh{(\beta \theta_t)} dt + \frac{dW_t}{\cosh{(\beta \theta_t)}} )$. The initial price is set to $P_{0} = 10$, and the volatility is defined as $\tilde{\sigma} = \eta \sigma$, with $\sigma = 0.1$. The imbalance is modeled as $q_t \sim \textit{Beta}\left( 1.5,1.5 \right)$. In each simulation, the Wiener process increments $dW_t$ are shared between both models. For $\beta=5$, we apply a scaling factor $\eta=1.45$ to achieve a comparable standard deviation of price changes---approximately 0.005 for both models---based on 1,000 simulations. For $\beta = 1$, we set $\eta = 1$. The time horizon is set to 1, with a simulation time step of$\frac{1}{N}$, where $N=390$, corresponding to 6.5 trading hours at 1-minute intervals.

\begin{figure}[!htbp]
    \centering
    \subfloat[\centering $\beta = 1$, $\eta=1$]{{\includegraphics[scale=0.4]{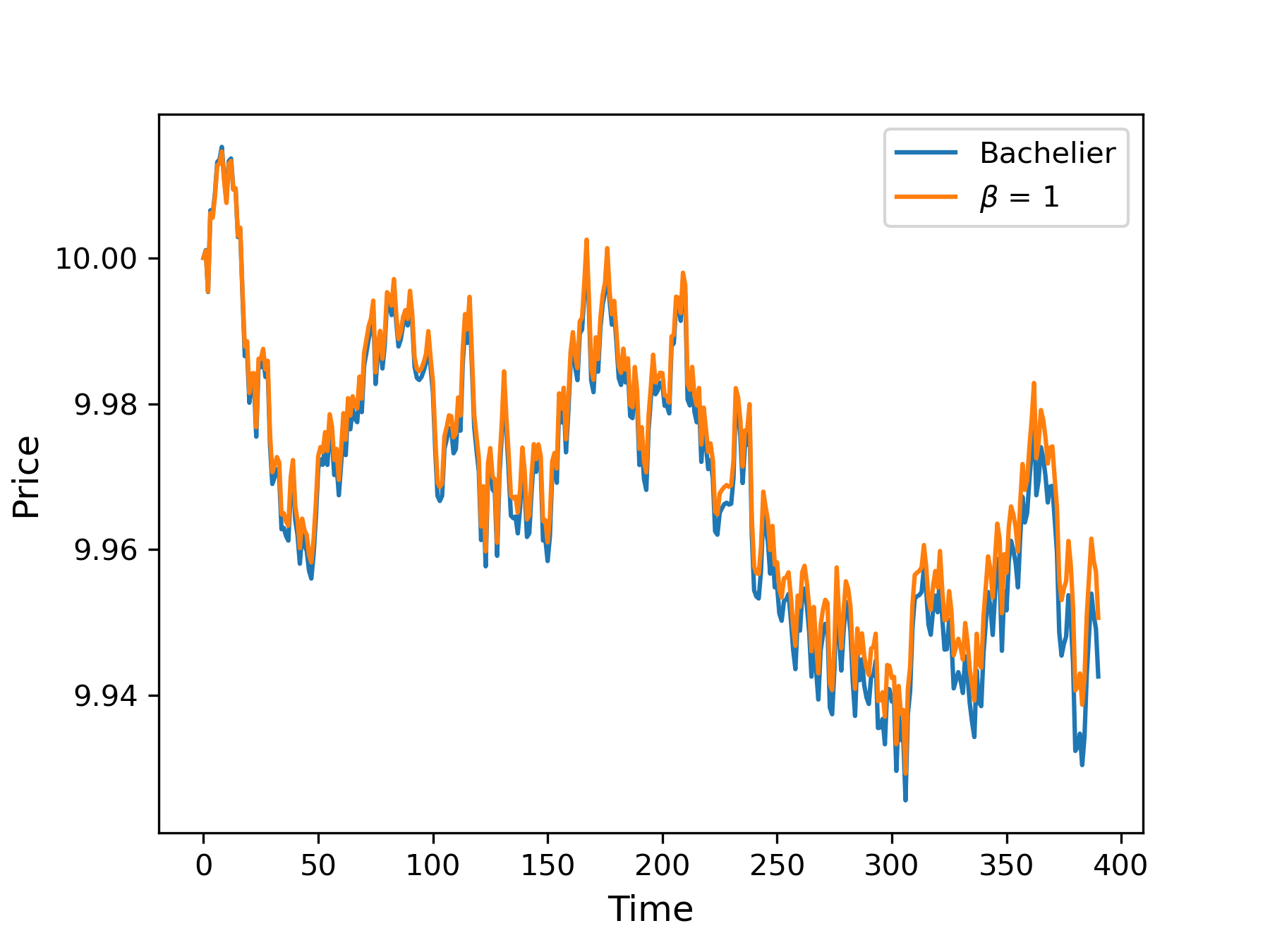} }}
    \qquad
    \subfloat[\centering $\beta = 5$, $\eta=1.45$]{{\includegraphics[scale=0.4]{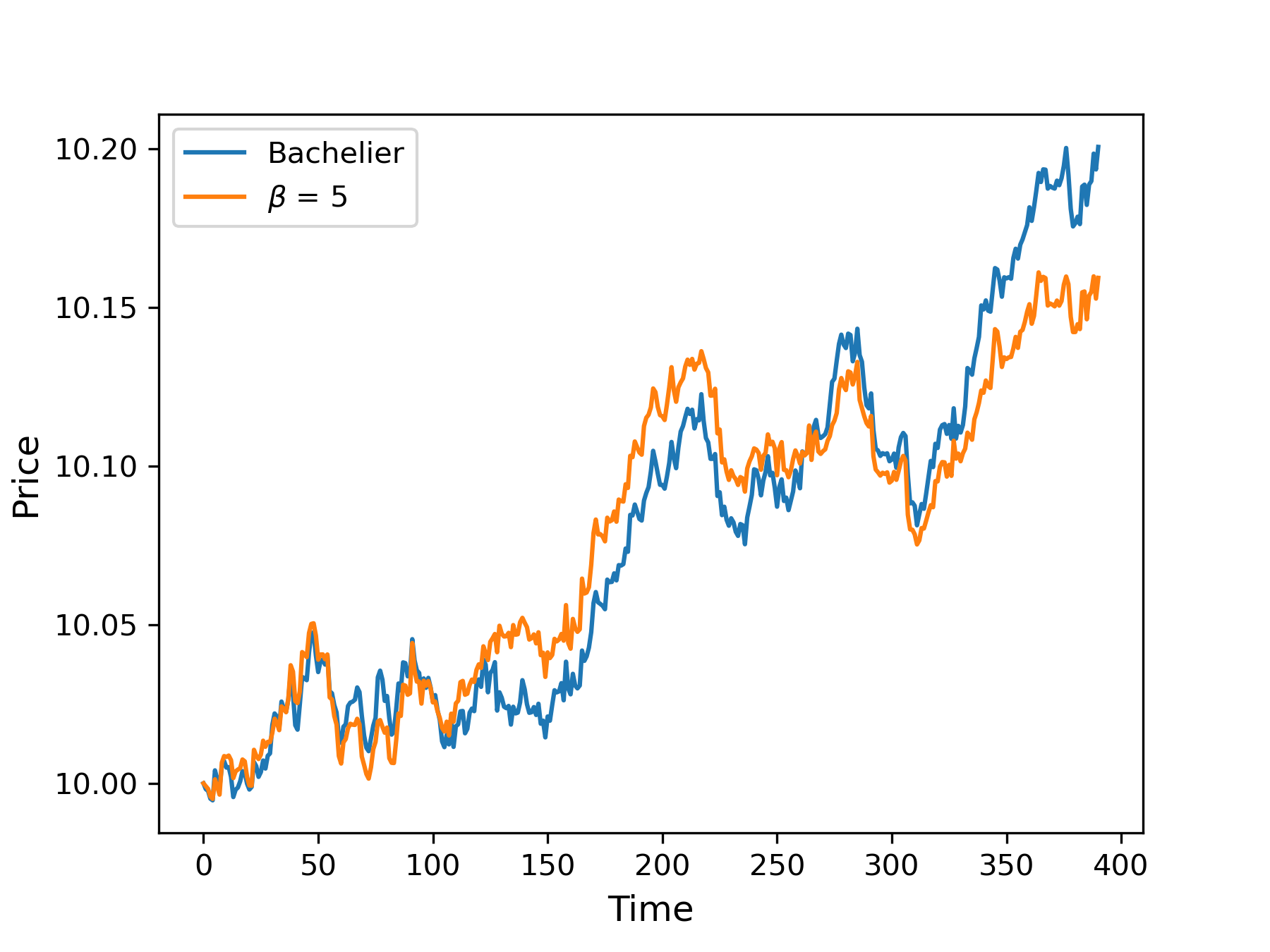} }}
    \caption{Price simulation: We compare the Bachelier model without drift, given by, $dP_t = \sigma dW_t$ to the price generated by Equation~\eqref{eq:price_dynamics}, i.e., $dP_t =  \tilde{\sigma} ( \tanh{(\beta \theta_t)} dt + \frac{dW_t}{\cosh{(\beta \theta_t)}} )$, where the initial price is $P_{0} = 10$, $\tilde{\sigma} = \eta \sigma$ the imbalance $q_t \sim \textit{Beta}\left( 1.5,1.5 \right)$, and $\sigma = 0.1$. In each simulation, the Wiener process increments $dW_t$ are shared between both models. For $\beta=5$, we apply a scaling factor $\eta=1.45$ to achieve a comparable standard deviation of price changes---approximately 0.005---based on 1,000 simulations. The time horizon is set to 1, with a simulation time step of$\frac{1}{N}$, where $N=390$, corresponding to 6.5 trading hours at 1-minute intervals.}
    \label{fig:price_simulation_beta1_5}
\end{figure}


As illustrated in Figure~\ref{fig:price_simulation_beta1_5}, we observe that when the sampled imbalance is symmetric and particularly concentrated around the center at $\frac{1}{2}$, both the Bachelier model without drift and the price process described by Equation~\eqref{eq:price_dynamics} are closely aligned for $\beta = 1$. However, as $\beta$ increases, the impact of imbalance becomes more pronounced and both curves exhibit greater variability, as seen for $\beta = 5$.

\section{Historical data}


In this section, we examine how well the proposed Equations~\eqref{eq:price_dynamics} and \eqref{eq:price_dynamics_varying_spread} aligns with historical data and analyze whether heavy-tailed distributions can also be modeled using these equations.

For historical data analysis, we used data obtained from \textit{Refinitiv Eikon}. We consider two types of stocks: one listed on the \textit{NYSE} with a varying spread, and another listed on the \textit{NASDAQ} with an effectively constant spread.

The trading hours were from 9:30 to 16:00, and we aggregated the data into 1-minute intervals using the last available quote within each bin. The dataset spans one month of trading in May 2025. We adopt a similar approach to \cite{HAGSTROMER2021} by excluding the beginning and end of each trading day, removing the first and last four minutes to reduce opening and closing effects.
We then compute mid-price changes and exclude the first available data point (9:35) each day to avoid overnight changes.\footnote{We exclude from the raw data any quotes for which $P^{b}=P^{a}$.}

The final sample consists of approximately 8,000 observations across 21 trading days.





\subsection{Stock Example with Varying Spread}\label{ssec:GE_N}

In this section, we analyze data from General Electic (GE) listed on the NYSE. To model price changes, we apply the dynamics described in Equation~\eqref{eq:price_dynamics_varying_spread}, incorporating a varying spread. Imbalance is modeled using a Beta distribution, and the spread is modeled using a Gamma distribution, following the approach in \cite{Lo_2002}. Figures~\ref{fig:hist_imb_vs_sampled_imb_GE_N} and~\ref{fig:hist_spread_vs_sampled_spread_GE_N} in the Appendix present a comparison between the sampled and historical data.\footnote{\label{sampled_vs_historical}It is worth noting that similar conclusions can naturally be drawn when using historical data in place of sampled spread and/or imbalance.}

The dataset spans 21 trading days and contains approximately 8,000 observations. To model the price dynamics, we set $\beta = 1$ and generate the sampled price differences as follows:

\begin{equation}
S_{i} \tanh{(\beta\theta_i)} + \eta S_{i} \frac{\epsilon_i}{\cosh{(\beta \theta_i)}}    
\end{equation}

where $\epsilon_i$ are i.i.d. random variables drawn from $\mathcal{N}(0,1)$, and $\eta=0.75$. Usually, it is necessary to scale either the drift term $dt$, the volatility term $dW_t$, or both. In our case, scaling the second term yielded a mean and standard deviation similar to those reported in \Cref{tab:mean_and_std_GE_N}.\footnote{We leave the estimation of the parameters in Equation~\eqref{eq:price_dynamics} for future research.}

\begin{table}[!htbp]
    \begin{minipage}{.48\linewidth}
      \centering
        \begin{tabular}{l c}
\hline
Statistic & Kurtosis \\
\hline
Mean      & 2.51 \\
Std       & 0.41 \\
Min       & 1.56 \\
Max       & 8.25 \\
\hline
\\
\end{tabular}
        \caption{Summary statistics of kurtosis values from
                $1,000$ runs with $\beta=1$ using sampled data.}
                \label{tab:kurtosis_GE_N_1000_summary}
    \end{minipage}%
    \hspace{0.04\linewidth}
    \begin{minipage}{.48\linewidth}
      \centering
        \begin{tabular}{l cc}
        \\
        \\
\hline
Statistic & $\Delta P^{\textit{mid}}$ & $\Delta P$\\
\hline
Mean      & 0.0046 & 0.005 \\
Std       & 0.131 & 0.1385 \\
\hline
\\
\end{tabular}
        \caption{Comparison of mean and standard deviation:
        mid-price changes vs. sampled price changes.
        }                
        \label{tab:mean_and_std_GE_N}
    \end{minipage}
\end{table}

\begin{figure}[!htbp]
    \centering
    \includegraphics[scale=0.5]{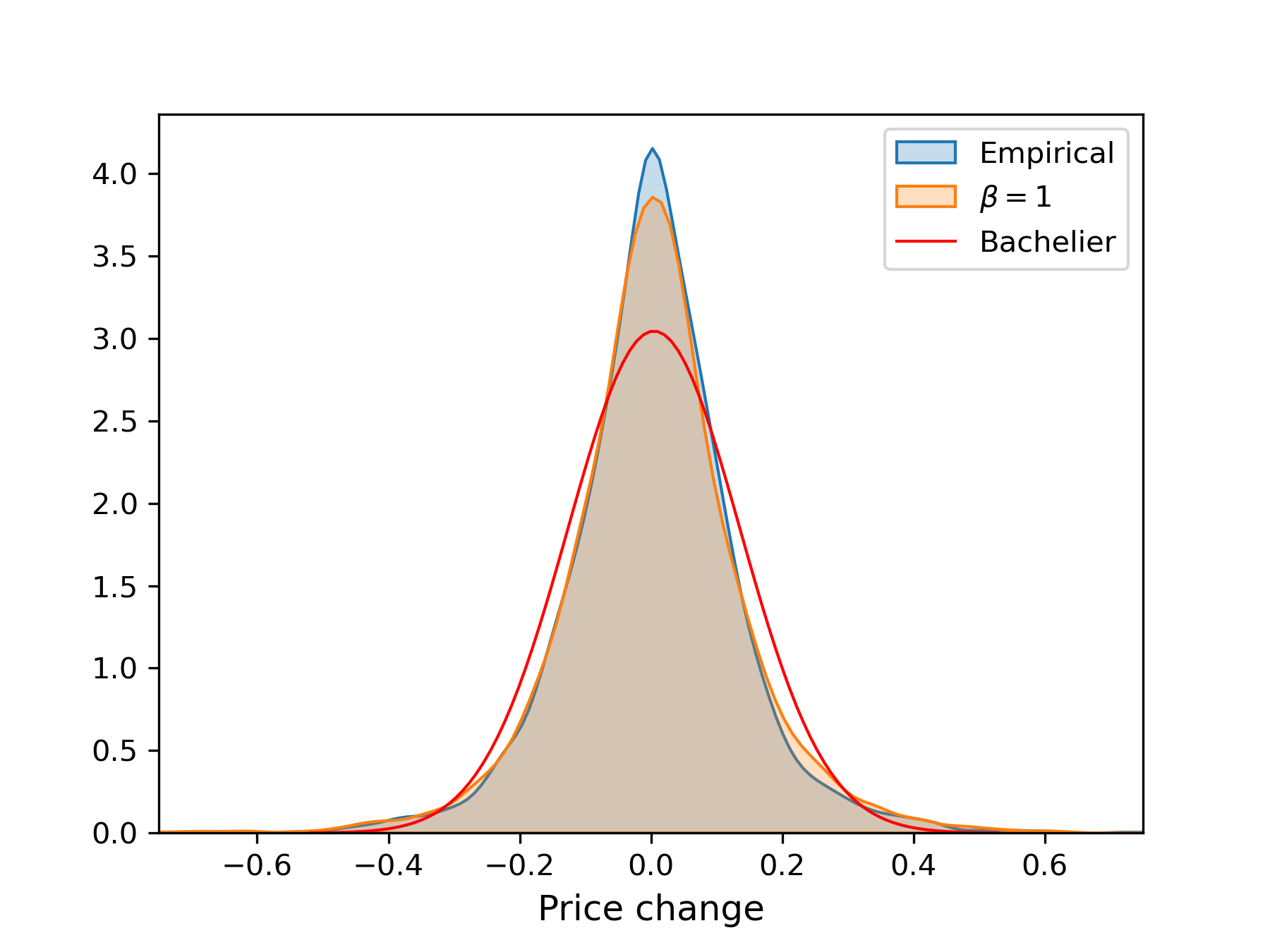}
    \caption{Kernel Density Estimation of mid-price changes and sampled price changes.
    The sampled price changes exhibit a kurtosis of 5.17, compared to 4.15 in the historical data.
    The theoretical density for the Bachelier model is also shown for comparison.}
    \label{fig:kde_density_GE_N_sampled}
\end{figure}

Using Kernel Density Estimation (KDE), we compared the distributions of mid-price changes and sampled price differences. As shown in Figure~\ref{fig:kde_density_GE_N_sampled}, the two distributions are closely aligned. More importantly, the kurtosis in the sampled data is consistent with that of the historical data. In the sampled data, kurtosis may vary due to differences in the sampled spread, imbalance and the normal random component. The sampled price changes displayed in Figure~\ref{fig:kde_density_GE_N_sampled} exhibit a kurtosis of 5.17, compared to 4.15 in the historical data. Refer to \Cref{tab:kurtosis_GE_N_1000_summary} for the corresponding statistics.



\subsection{Stock Example with Constant Spread}\label{ssec:LCID_O}

Small-cap stocks typically exhibit larger bid-ask spreads than large-cap stocks; see, e.g., \cite{Chordia_2001}. Moreover, the return distributions of small-cap stocks have been shown to have significantly fatter tails than those of large-cap stocks \cite{Amihud_2002, Eom_2019}. This pattern is also observed in our examples.\footnote{GE has a much higher market cap than LCID. However, the average mid-price in our sample for LCID is approximately \$2.57, compared to about \$225.25 for GE. LCID is a very liquid stock with a high trading volume. This combination of high liquidity and low price results in an effectively constant spread.} More importantly, modeling price changes with heavier tails presents additional challenges.

To analyze a stock with a stable and nearly constant spread, we use Lucid Group (LCID), which is listed on NASDAQ. As before, we model price changes using the dynamics described in Equation~\eqref{eq:price_dynamics}, assuming a constant spread. For LCID in our 1-minute interval data, the spread is consistently 0.01. The imbalance is again modeled using a Beta distribution.
Figure~\ref{fig:hist_imb_vs_sampled_imb_LCID_O} in the Appendix presents two histograms comparing the sampled and historical imbalance data.\cref{sampled_vs_historical}

As shown in Figure~\ref{fig:hist_weighted_vs_boltzmann_change_LCID_O}, the mid-price typically remains unchanged or changes by $\pm 0.01$. We compare the Boltzmann price changes with parameters $\beta =0.5$, $1$, and $2$ to the mid-price changes. The Boltzmann price changes with $\beta = 0.5$ closely resemble the mid-price changes, while the weighted mid-price changes appear noticeably noisier. In contrast, the Boltzmann price changes with $\beta = 1$ offer more informative signals than the mid-price changes, while remaining less noisy than the weighted mid-price changes. All price change distributions exhibit elevated kurtosis.
We computed a kurtosis of 6.99 for the mid-price changes, approximately 7.54 for the Boltzmann price changes with $\beta = 0.5$, 8.08 for the Boltzmann price with $\beta=1$, and 9.1 for the weighted-mid price changes.

\begin{figure}[htb]
    \centering
    \subfloat[\centering The mid-price vs the weighted-mid price changes]{{\includegraphics[scale=0.4]{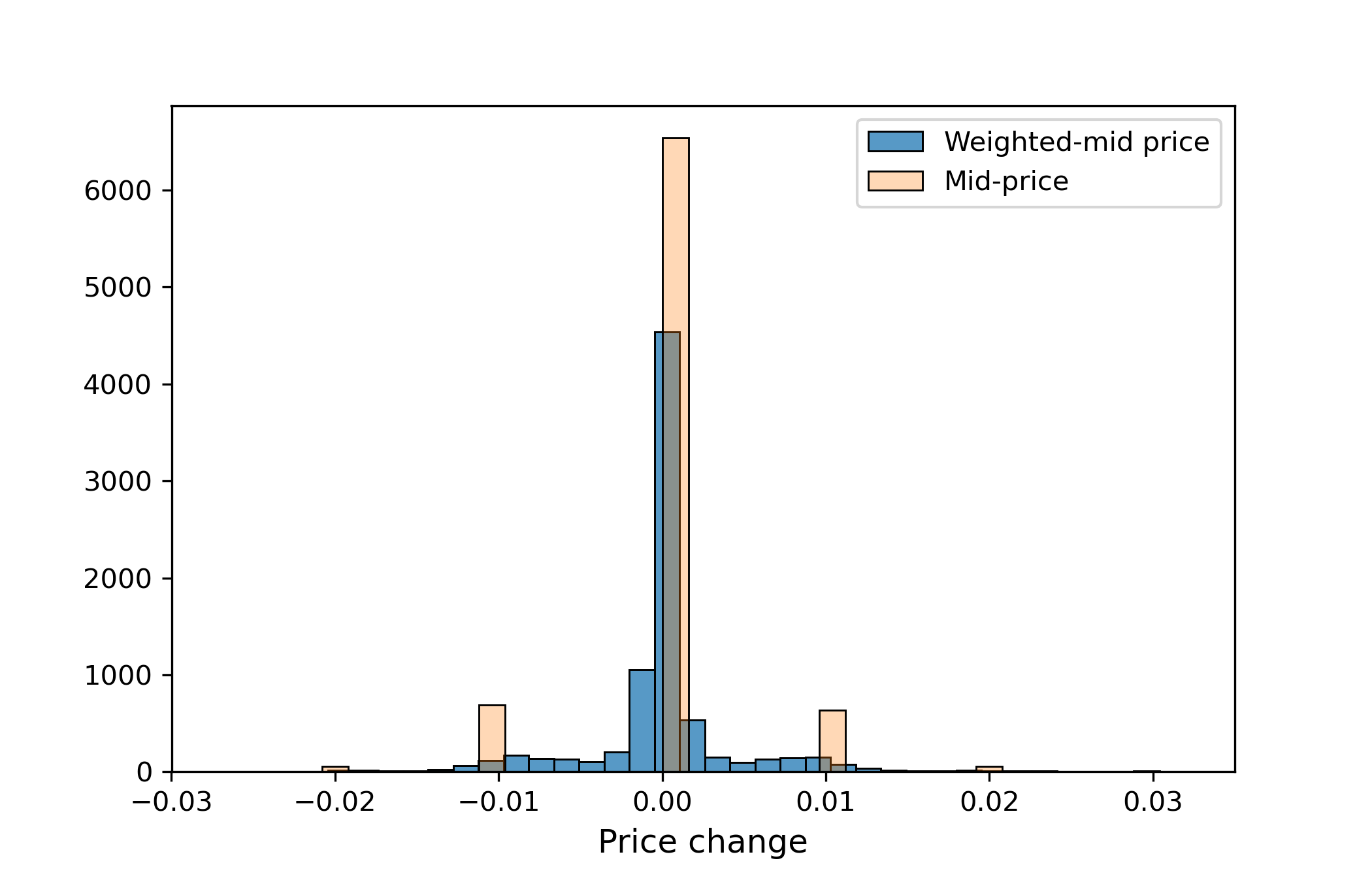} }}
    \qquad
    \subfloat[\centering The mid-price vs Boltzmann price ($\beta=0.5$) changes]{{\includegraphics[scale=0.4]{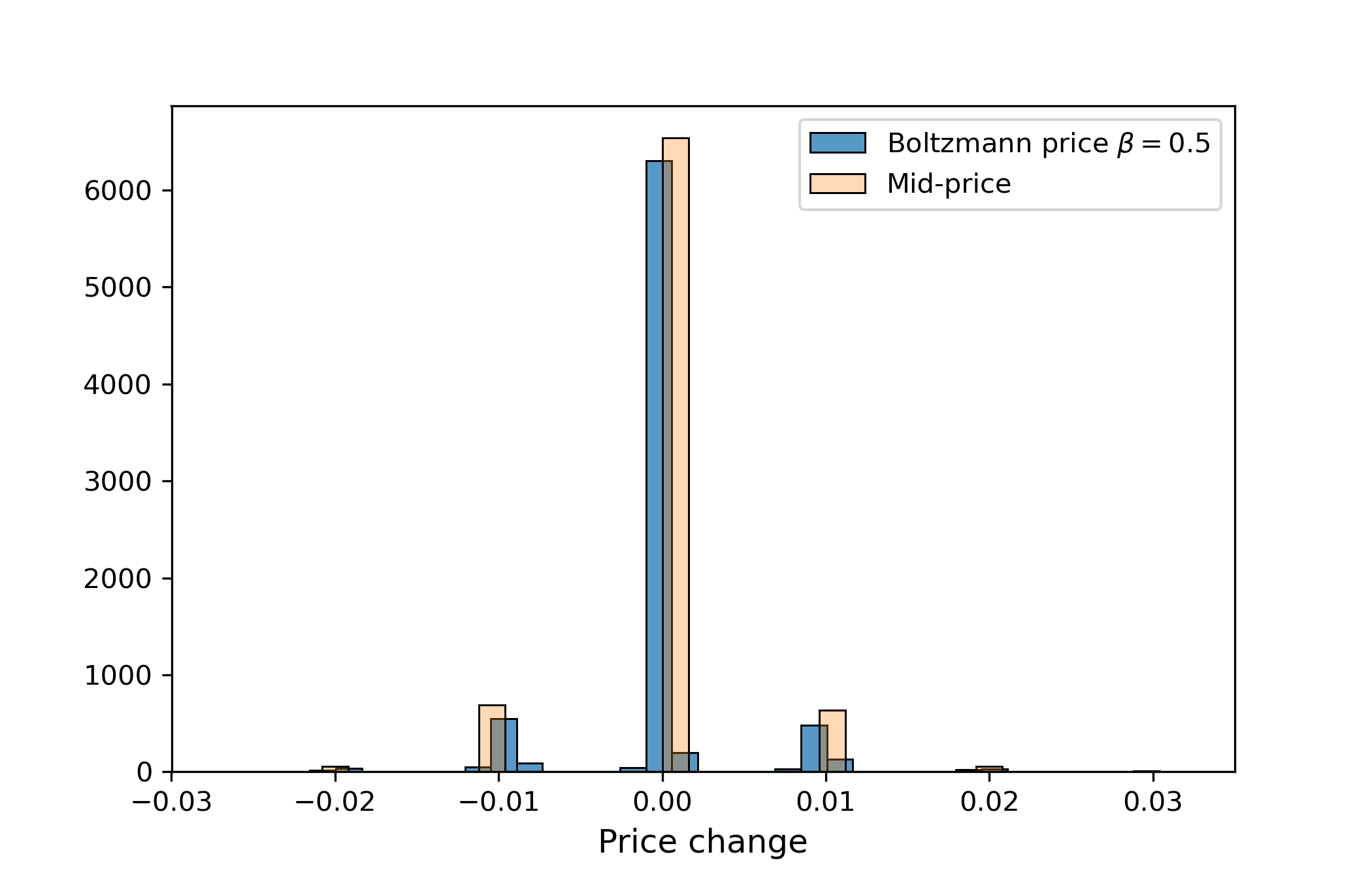} }}
    \qquad
    \subfloat[\centering The mid-price vs equilibrium price changes]{{\includegraphics[scale=0.4]{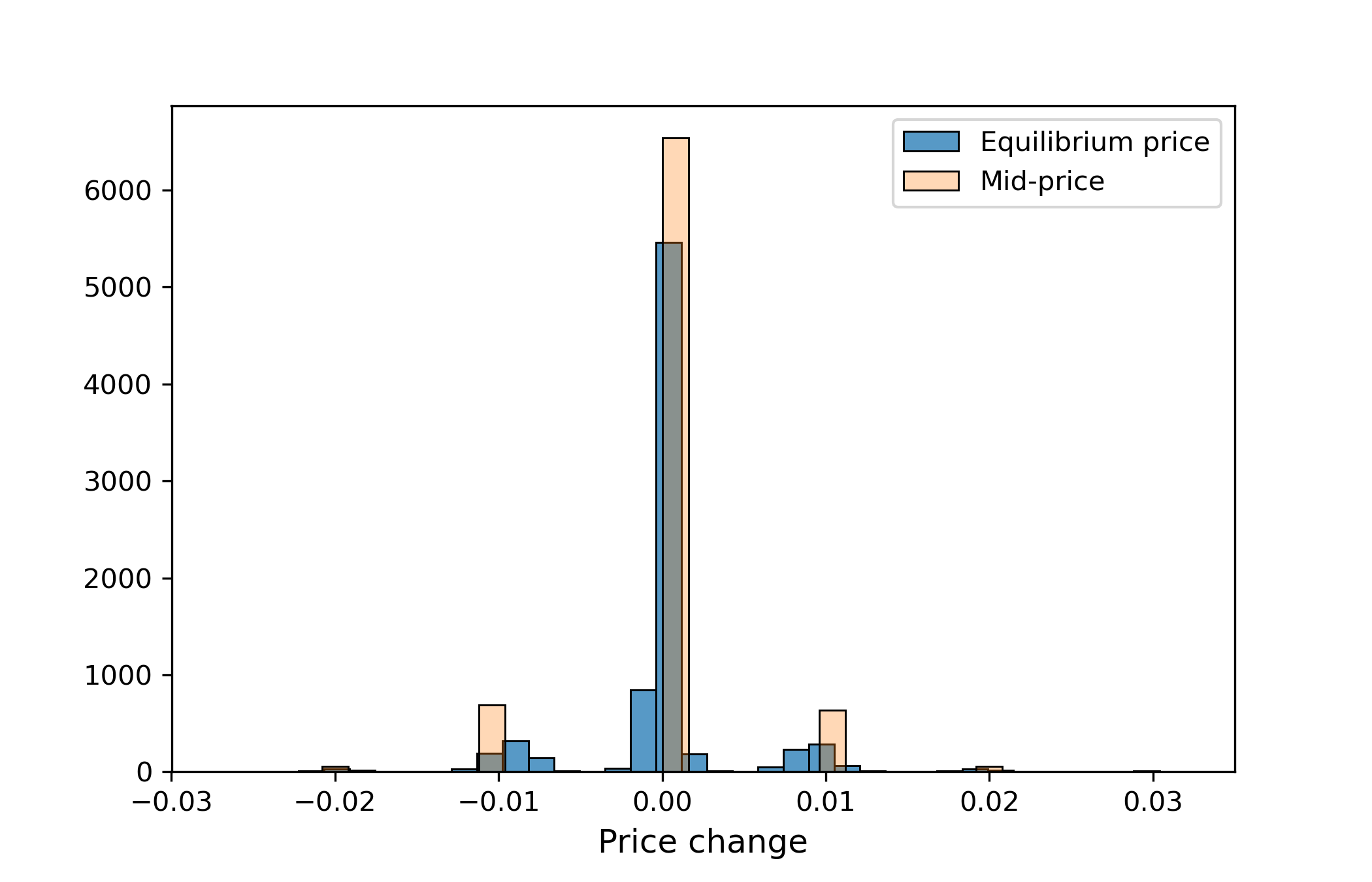} }}
    \caption{Histograms for LCID, based on data aggregated into 1-minute intervals. The Boltzmann price changes for $\beta =0.5$, $1$, and $2$ are compared with the mid-price changes.
    }
    \label{fig:hist_weighted_vs_boltzmann_change_LCID_O}
\end{figure}



The dataset spans 21 trading days and contains approximately 8,000 observations.
Following the example in Subsection~\ref{ssec:sim_const_spread}, we increase the impact of imbalance by setting $\beta = 17$. The sampled price differences are generated as follows:

\begin{equation}\label{const_spread_sampling}
\eta \tilde{\mu} \tanh{(\beta\theta_i)} + \eta \tilde{\sigma} \frac{\epsilon_i}{\cosh{(\beta \theta_i)}},
\end{equation}

where $\epsilon_i$ are i.i.d. random variables drawn from $\mathcal{N}(0,1)$, and $\tilde{\mu}$ and $\tilde{\sigma}$ represent the mean and standard deviation of mid-price changes, respectively. We observed that using $\beta = 17$ and $\eta = 2$ yields satisfactory results; however, it is important to note that more accurate estimates may be obtained. In \Cref{tab:kurtosis_LCID_O_1000_summary}, we compare the statistics from two runs with different parameter settings. Each run consists of 1,000 samples, with each sample comprising price changes generated according to formula~\ref{const_spread_sampling}. The results are quite close, but ultimately, choosing a higher $\beta$ appears to fit the data slightly better.

Generally, we observe that increasing $\beta$ increases kurtosis but, on the other hand, decreases the variance of the sampled price changes. The parameters $\beta$ and $\eta$ were selected to produce sampled price changes with means and standard deviations similar to those in the historical data, while also achieving comparable kurtosis. Specifically, the average standard deviation across 1,000 samples is nearly identical to that of the historical price changes: 0.00477 for the sampled series and 0.00487 for the historical series, respectively.

For visualization, we also present the distribution and summary statistics of a single sample whose kurtosis closely matches that of the historical mid-price changes. As shown in \Cref{tab:mean_and_std_LCID_O}, the mean and standard deviation of the corresponding price changes are comparable.

\begin{table}[htb]
    \begin{minipage}{.5\linewidth}
      \centering
        \begin{tabular}{l cc}
\hline
Statistic & $\beta=8.5$ & $\beta=17$\\
 & $\eta=1$ & $\eta=2$\\
\hline
Mean      & 5.5 & 5.62\\
Std       & 0.29 & 0.38\\
Min       & 4.75 & 4.29\\
Max       & 6.61 & 7.56\\
\hline
\\
\end{tabular}
        \caption{Summary statistics of kurtosis values\\
        from 1,000 samples.}
                \label{tab:kurtosis_LCID_O_1000_summary}
    \end{minipage}%
    \begin{minipage}{.5\linewidth}
      \centering
        \begin{tabular}{l ccc}
        \\
        \\
\hline
Statistic & $\Delta P^{\textit{mid}}$ & $\Delta P$ & Rounded $\Delta P$\\
\hline
Mean      & \num{-3.63e-5} & \num{-4.54e-5} & \num{-2.88e-5} \\
Std       & 0.0049 & 0.0048 & 0.005\\
\hline
\\
\end{tabular}
        \caption{Comparison of mean and standard deviation: mid-price changes vs. sampled and rounded sampled price changes for a representative sample.}
        \label{tab:mean_and_std_LCID_O}
    \end{minipage}
\end{table}

To better match the real data distribution we round the sampled data to two decimal places.\footnote{There is extensive research on price rounding and its effects; see, e.g. \cite{Ball_2001}.} The kurtosis of the sampled price changes is 6.74, compared to 6.99 in the historical data. After rounding the sampled price changes to two decimal places, the kurtosis increases marginally to 6.84.
In practice, kurtosis may vary due to randomness of the sampled imbalance and the normal component.


In Figure~\ref{fig:hist_LCID_O_sampled_beta_17}, we compare the histograms of the sampled price changes and historical mid-price changes , which shows a close match between the two distributions. Importantly, the kurtosis of the sampled data is consistent with that of the real data. Refer to \Cref{tab:kurtosis_LCID_O_1000_summary} for detailed statistics.

\begin{figure}[htbp]
    \centering
    \subfloat[\centering Histogram of sampled price changes (kurtosis 6.74).]{{\includegraphics[scale=0.42]{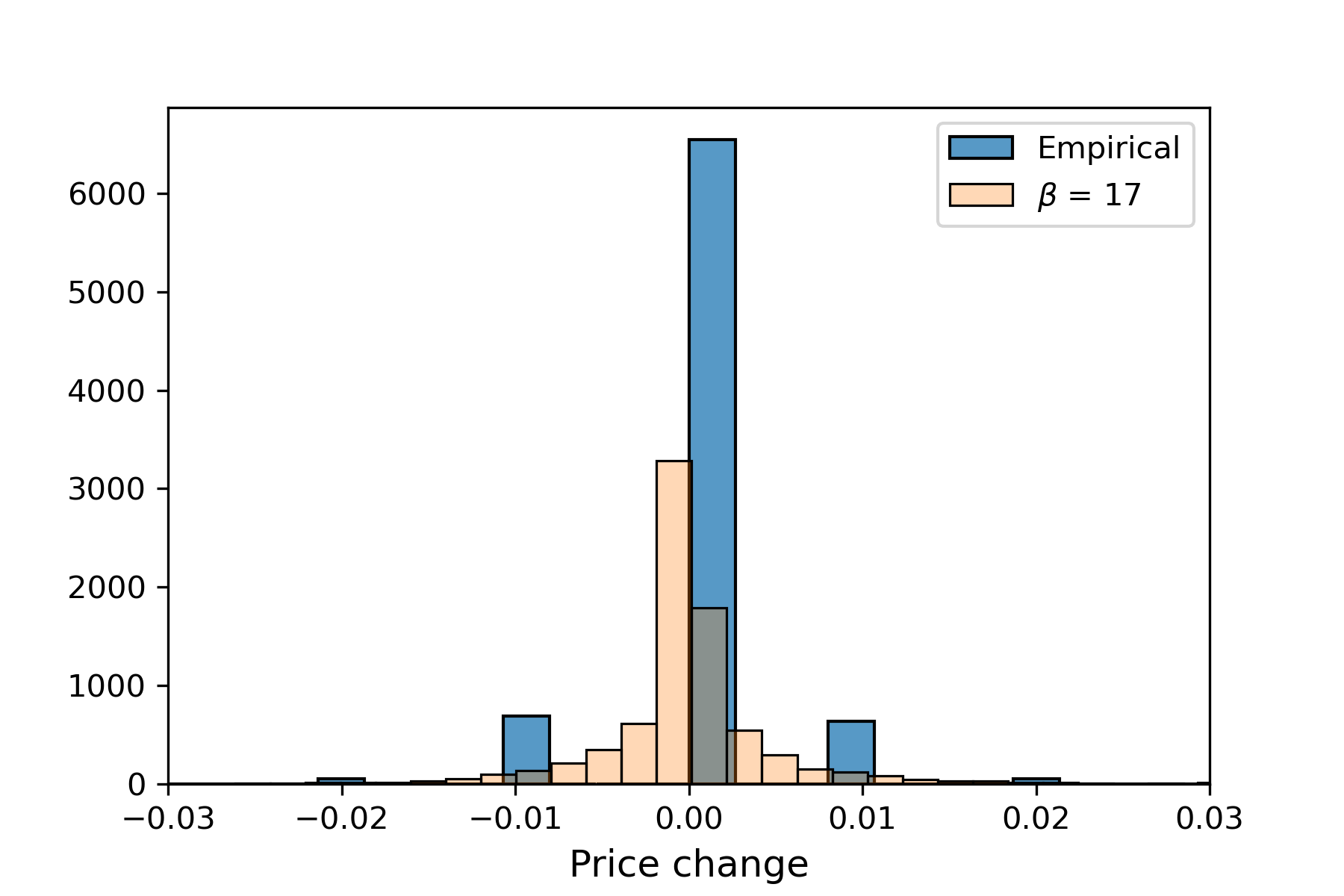} }}%
    \qquad
    \subfloat[\centering Histogram of rounded sampled price changes (kurtosis 6.84).]{{\includegraphics[scale=0.42]{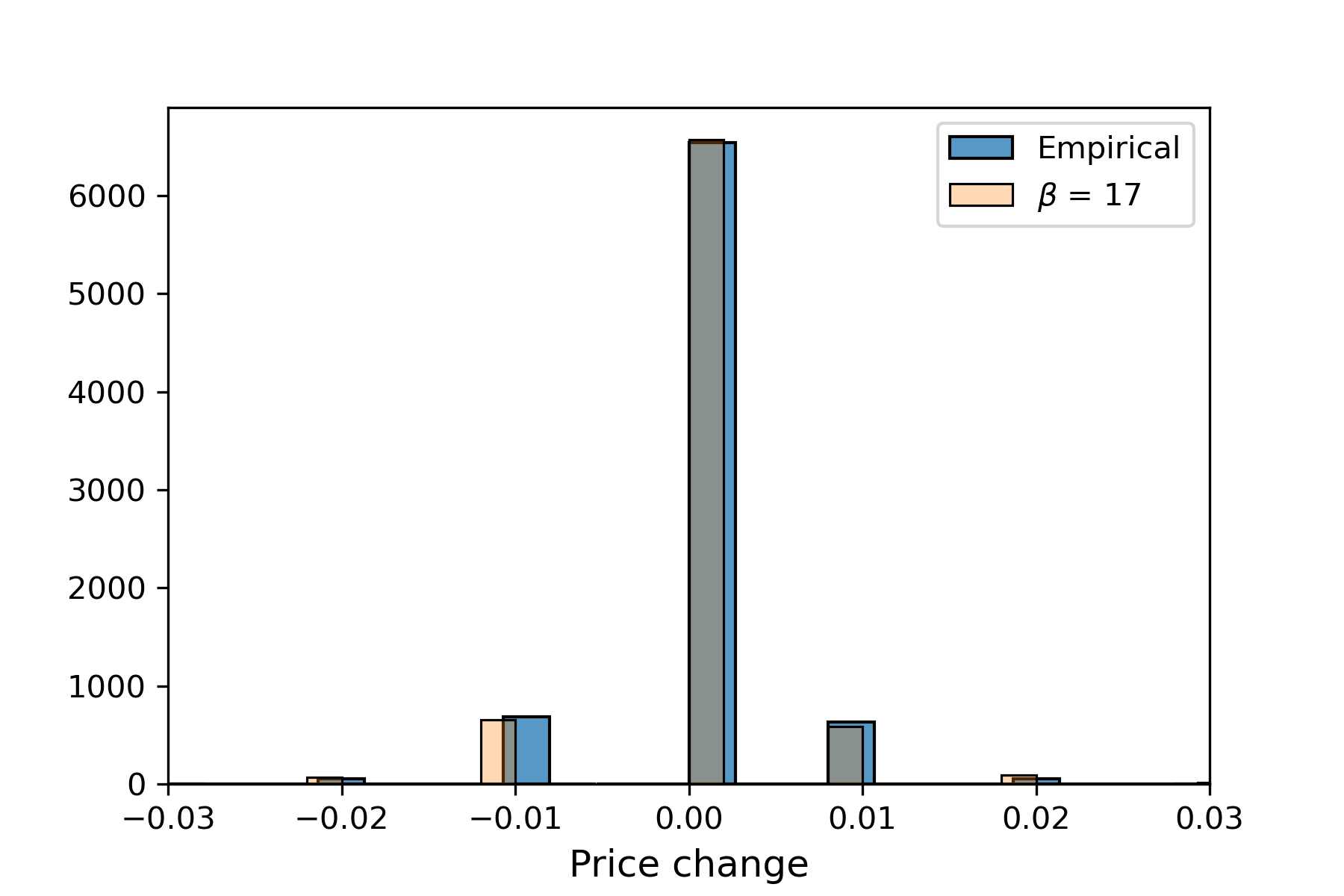} }}%
    \caption{Comparison of mid-price changes (historical) vs. sampled price changes, before and after rounding.
    In the sampled data, kurtosis may vary due to randomness in the sampled imbalance and the normally distributed noise component. The kurtosis of the sampled price changes is 6.74, compared to 6.99 in the historical data. After rounding the sampled price changes to two decimal places, the kurtosis increases marginally to 6.84. Sampling we conducted using $\beta = 17$.}
    \label{fig:hist_LCID_O_sampled_beta_17}
\end{figure}

\section{Market Impact}


The theoretical model of market impact was proposed by Kyle \cite{Kyle_1985} where the model assumes three kinds of traders, i.e., a single insider (informed trader), uninformed noise traders, and market makers. Especially market makers use information about the quantities traded by other market participants. In Kyle's model the optimal strategy of informed trader is to split the metaorder into smaller pieces and execute them incrementally over time, preventing reveiling the information to early.
The price is adjusted linearly by market makers based on the information of the order flow, defined as the current and past aggregate quantities traded by the insider and noise traders. The rate of this linear growth is called Kyle's $\lambda$.

The study of fundamental price is crucial for market impact analysis, where it is typically modeled relative to the mid-price. In the literature, the first generation of market impact models \cite{Bertsimas_1998, Almgren_1999, Almgren_2001, Almgren_2003} commonly distinguishes between permanent and temporary impacts. In particular, there is no resilience in the model of permanent market impact. It was shown that assuming no dynamic arbitrage, only a permanent impact which is a linear function of volume is allowed \cite{Huberman_2004,Gatheral_2010}.

According to Almgren and Chriss \cite{Almgren_2001}, temporary market impact arises from temporary imbalances in supply and demand caused by trading, leading to short-lived price movements away from equilibrium. Following \cite{Almgren_2001,Cartea_2015}, the execution price $\hat{P}_t$ at which the agent can sell or purchase the asset at time $t$ is given by:
\begin{equation}
    \hat{P}_t = P_t^{\textit{mid}} \pm \left( \frac{\textit{S}}{2} + h(\nu_t) \right),
\end{equation}
where $S$ is the bid-ask spread (assumed constant), and $h(\nu_t)$ represents the temporary market impact, with $h(\nu_t) \geq 0$. The sign $\pm$ is positive $(+)$ for buy orders and negative $(-)$ for sell orders, reflecting the direction of the trade. Almgren and Chriss \cite{Almgren_2001} propose a linear form of temporary market impact, $h(\nu_t) = \eta \nu_t$, where $\eta$ is a constant parameter.\footnote{In \cite{Almgren_2001}, the temporary market impact includes the spread term $\frac{\textit{S}}{2}$, but we follow the notation of \cite{Cartea_2015} for consistency.} However, empirical evidence suggests that $h(\nu_t)$ is better described by a power-law function (see, e.g., \cite{Cartea_2015}).

Half the spread, $\frac{\textit{S}}{2}$, represents the transaction cost of crossing the spread to execute a market order. At any time $t$, the number of shares available at the bid and ask prices, $P_t^{\textit{mid}} \pm \frac{S}{2}$, is constrained by the liquidity at Level I of the LOB. When a large order is executed, it consumes liquidity at higher price levels in the LOB, resulting in an execution price $\hat{P}_t$ worse than the quoted bid or ask prices. This additional cost, beyond half the spread, is modeled by the temporary market impact $h(\nu_t)$. This temporary impact affects $\hat{P}_t$ without altering the mid-price process $(P_t^{\textit{mid}})$. It is assumed that this effect is short-lived and liquidity is restored shortly after execution, and a new equilibrium price is established.\cite{Almgren_2001}

The execution cost represented by half the spread, $\frac{S}{2}$, is a fundamental concept in market microstructure literature \cite{Glosten_1985}. It reflects the fixed cost of executing a market order relative to the mid-price. Specifically, a buyer pays the ask price ($P^{\textit{mid}} + \frac{S}{2}$), while a seller receives the bid price ($P^{\textit{mid}} - \frac{S}{2}$).

We should note that more recent studies emphasize the transient nature of market impact of single market orders in the sense that it decays over a certain time interval after trade \cite{Bouchaud_2009,Bacry_2014}.

This section examines the impact of small (infinitesimal) market orders on the fundamental price and presents the general dynamics of the market impact of such orders.

Here, we try to justify the cost given by half the spread $\frac{S}{2}$ by the change in the fundamental price. Moreover, we challenge this value; one might notice that when buy order occurs, it captures the liquidity on ask side.  It should push the fundamental price towards the ask side, changing the cost.

Let us assume that the imbalance $q$ has changed by $\delta$, resulting in the Boltzmann price change. Let $q = \frac{1}{2} + \theta$ The derivative of $P^{\textit{boltzmann}}$ with respect to $\theta$ is equal to:\footnote{Notice that $\left(P^{\textit{boltzmann}}\right)^{\prime}\left( q \right) = \frac{2\beta e^{-\beta} \left( P^a - P^b \right)}{\left( e^{-\beta q} + e^{-\beta (1-q)} \right)^2}$, but using $\theta = q - \frac{1}{2}$ simplifies the calculation.}
\begin{equation}
    \left(P^{\textit{boltzmann}}\right)^{\prime}\left( \theta \right) = \frac{2\beta \left( P^a - P^b \right)}{\left( e^{-\beta \theta} + e^{\beta \theta} \right)^2} = \beta \cdot \frac{\frac{S}{2}}{\cosh^2{\left( \beta \theta \right)}}.
\end{equation}
Hence, the dynamics of the Boltzmann price is given by
\begin{equation} \label{eq:boltzmann_market_impact_dynamics}
    d P^{\textit{boltzmann}} = \beta \cdot \frac{\frac{S}{2}}{\cosh^2{\left( \beta (q-0.5) \right)}} dq.
\end{equation}
This dynamics reminds the results from literature where the impact is proportional to half the spread, compare, e.g., \cite{Madhavan_1997, Wyart_2008}. However, in Equation~\eqref{eq:boltzmann_market_impact_dynamics} the change in fundamental price depends also on bid and ask imbalances providing the impact is slightly less than half the spread.

Using Lemma \ref{prop:price_approx} and assuming $\beta = 1$ we can rewrite the above equation as follows
\begin{equation}
    d P^{\textit{eq}} \approx \frac{S}{2} dq.
\end{equation}
Suppose that the imbalance $\theta$ has changed between times $t_0$ and $t_1$ by $\Delta \theta = \theta_1 - \theta_0$. The price change can be calculated, using the \textit{Fundamental Theorem of Calculus}, as
\begin{equation}
    \Delta P^{\textit{boltzmann}} = P^{\textit{boltzmann}}_{t_1} - P^{\textit{boltzmann}}_{t_0} = \beta \cdot \frac{S}{2} \int_{\theta_0}^{\theta_1} \frac{1}{\cosh^2{\left( \beta \theta \right)}} d\theta = \beta \frac{S}{2} \frac{1}{\beta} \tanh{\left( \beta \theta \right)} \Biggr|_{\theta_0}^{\theta_1}.
\end{equation}
Hence,
\begin{equation}
    \Delta P^{\textit{boltzmann}} = \frac{S}{2} \cdot \left[ \tanh{\beta\theta_1} - \tanh{\beta\theta_0} \right].
\end{equation}
Now imagine that we have a completely balanced TOB, i.e. $q=\frac{1}{2}$ and we completely clear, e.g., the ask side, resulting in $q=1$. Let us calculate the Boltzmann price change. We have:
\begin{equation}
    P^{\textit{boltzmann}}\left( \theta=\frac{1}{2} \right) - P^{\textit{boltzmann}}\left( \theta=0 \right) = \frac{S}{2} \cdot \left[ \tanh{\left( \frac{\beta}{2} \right)} - \tanh{\left( 0 \right)}\right] = \frac{S}{2} \cdot \tanh{\left( \frac{\beta}{2} \right)}
\end{equation}
Let $\beta = 1$, i.e., we consider the equilibrium price $P^{\textit{eq}}$. The price has then changed by $\frac{S}{2} \cdot \tanh{(0.5)}$ where $\tanh{(0.5)} \approx 0.462$. Analogously, we get the same results but with opposite sign when we sell shares. For example, moving from $\theta=0$ towards $\theta = -\frac{1}{2}$, results in a price change by $-\frac{S}{2} \cdot \tanh{\left( \frac{\beta}{2} \right)}$.




Notice that the mid-price $P^{\textit{mid}}$ is fixed and does not depend on the imbalance. Therefore, its derivative is zero, and we cannot describe its dynamics in the same way as above. This naturally raises the question: what about the weighted-mid price---is it a good proxy for the fundamental price? Recall that
\begin{equation}
    P^{\textit{w}} = q^b P^a + q^a P^b = \left( \frac{1}{2} + \theta \right) P^a + \left( \frac{1}{2} - \theta \right) P^b.
\end{equation}
Calculating the derivative, we easily get
\begin{equation}
    \left( P^{\textit{w}} \right)^{\prime} (q) = \left( P^{\textit{w}} \right)^{\prime} (\theta) = P^a - P^b = S. 
\end{equation}
Hence, the dynamics for the weighted-mid price is given by:
\begin{equation}
    d P^{\textit{w}} = S dq.
\end{equation}
If a good measure of the fundamental price is expected to adjust by roughly half the spread, this serves as a strong argument against using the weighted mid-price as a proxy. It can also support the use of alternative measures, such as $\frac{1}{2}\left( P^{\textit{mid}} + P^{\textit{w}} \right) \approx P^{\textit{eq}}$, as more appropriate candidates.\footnote{Notice that the derivative with respect to $q$ is equal to $( \frac{1}{2}( P^{\textit{mid}} + P^{\textit{w}} ) )^{\prime} = \frac{S}{2}$.}

This means that in a situation where the TOB is balanced ($\theta=0$), and we clear (buy) all of the ask side at the given spread ($\theta=\frac{1}{2}$), the resulting price change is given by:
\begin{equation}
    \Delta P^{\textit{w}} = S \cdot \frac{1}{2} - S \cdot 0 = \frac{S}{2}.
\end{equation}
However, in reality, this does not necessarily mean that the next mid-price will move up by half the spread, because the LOB exhibits resilience, and new orders may arrive. There is also the effect of bid-ask bounce. All of this suggests that the impact modeled by the weighted mid-price might be overestimated. In contrast, the temporary market impact modeled by the Boltzmann price appears to better reflect the resilience of the LOB.


\begin{figure}[htb]
    \centering
    \subfloat[\centering The derivative of the Boltzmann price compared to the mid-price adjustment for $\beta = 1$.]{{\includegraphics[scale=0.45]{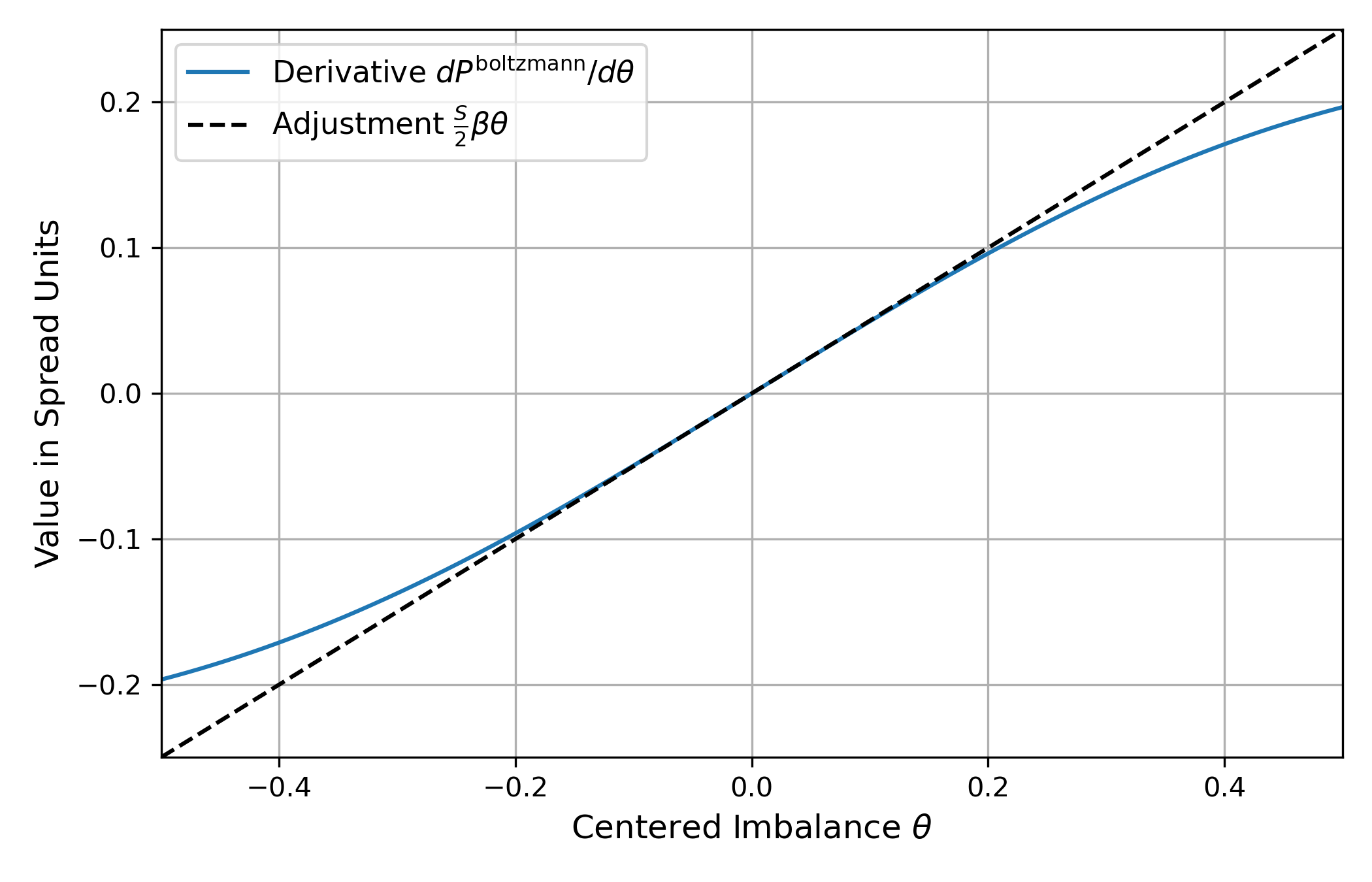} }} \label{fig:price_impact_comparison_a}
    \qquad
    \subfloat[\centering Price change $P^{\textit{boltzmann}}(\theta) - P^{\textit{boltzmann}}(0)$ as the theoretical impact when changing the centered imbalance from $0$ to $\theta$. The dashed lines represent the weighted mid-price and the mid-price changes respectively.]{{\includegraphics[scale=0.45]{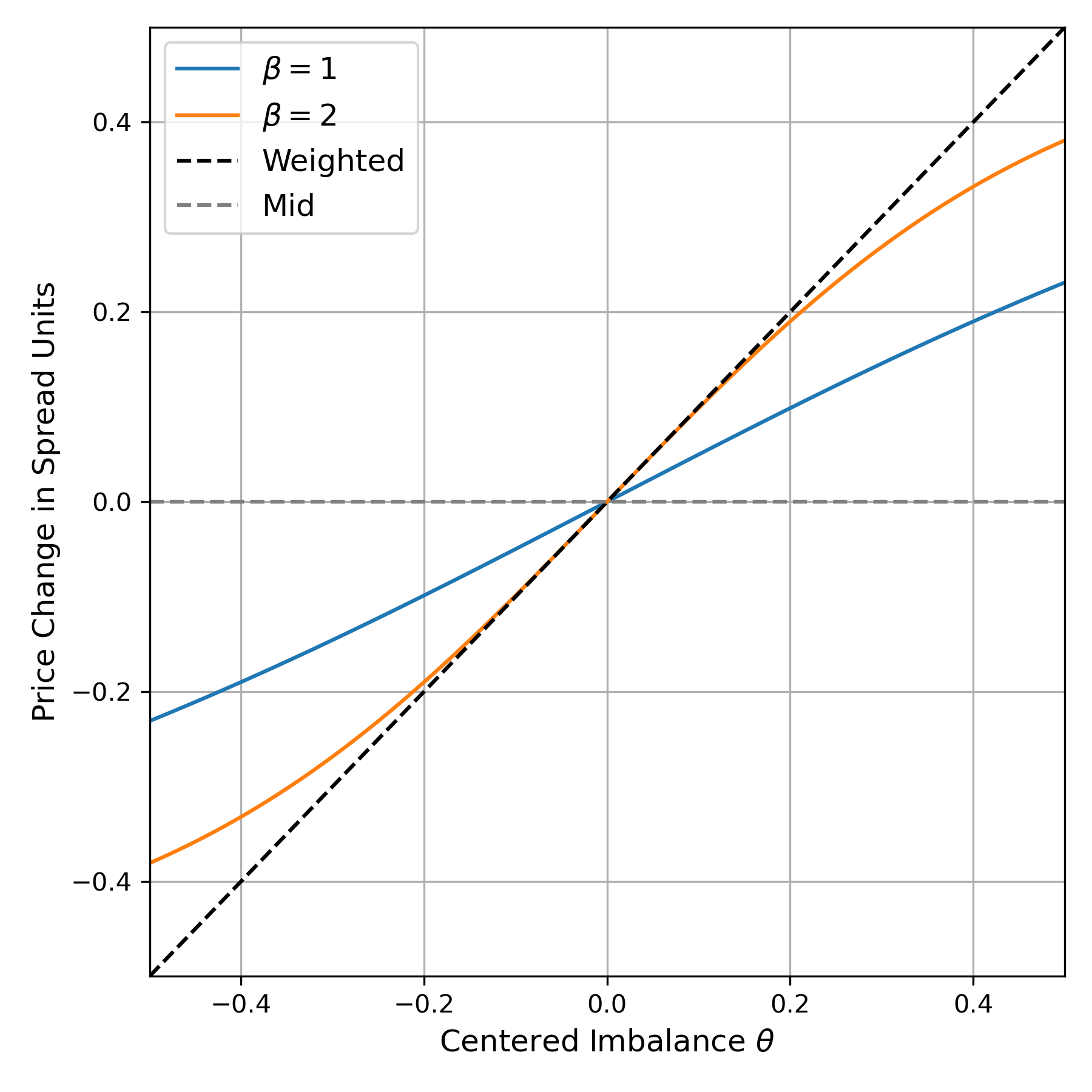} }} \label{fig:price_impact_comparison_b}
    \caption{The derivative of the Boltzmann price $\frac{d}{d\theta}P^{\textit{boltzmann}}$ vs. theoretical price change $\Delta P^{\textit{boltzmann}}$. The theoretical price move is the drift term, i.e., $\frac{S}{2} \tanh{\left( \beta \theta \right)}$. There has been observed, that the average price move can be up to a third of the spread \cite{Lipton2013}.
    }
    \label{fig:price_impact_comparison}
\end{figure}



Using the proposed Boltzmann price as the fundamental price and modeling market impact---e.g., as in Equation~\eqref{eq:boltzmann_market_impact_dynamics}---it becomes evident that not only trades influence fundamental price changes, but also any modifications to the bid and ask quotes.\footnote{At least at Level I. Higher levels of the LOB are not considered in this analysis.} This observation helps to partially explain the transient nature of market impact. Even when trades occur and affect the price, the order book continues to evolve, exhibiting resilience, and continuously seeking a new equilibrium. Although such dynamics may not be fully visible at the trade level, it becomes more apparent when observing the fundamental price.


\section{Conclusions and Future Research}

In this paper, we introduced the parametrized family of prices derived from the Maximum Entropy Principle. Both states bid and ask are assigned probabilities given by the Boltzmann distribution accounted for volume imbalances.

Using these probabilities, we derived the dynamics of a price process in which both the drift and volatility are affected by volume imbalance. We demonstrated that such dynamics can generate higher kurtosis and heavier tails compared to, for example, the Bachelier or Geometric Brownian motion models. In our examples, the modeling approach differed between stocks with variable spreads and those with a constant spread.

Building on the Boltzmann price concept, we considered a market impact model and showed that market impact can naturally emerge from the proposed price dynamics, generating a price drift.

Nevertheless, much remains to be explored and validated. A detailed empirical analysis using historical data---including trading prices across various time intervals---is needed to assess whether the proposed dynamics accurately capture real-world price behavior. A comprehensive study of model parameter estimation should be conducted, complemented by a theoretical analysis of the proposed equations. Additionally, future analyses should include not only price changes but also price returns and log returns.

It is valuable to focus on the GBM equivalent of our model, particularly for modeling longer time periods.

An especially promising idea is to analyze an estimator of the effective bid-ask spread using the Boltzmann price as a benchmark and to evaluate its bias across different values of the parameter $\beta$.

\bibliographystyle{unsrt}
\bibliography{references}  






\section{Appendix}


\subsection{Micro-Price} \label{sec:micro-price}

As an alternative to the mid-price and the weighted mid-price, Stoikov \cite{Stoikov_2018} proposed the micro-price as a proxy for the fundamental price. In brief, it can be described as the current mid-price adjusted for the expected future mid-price movements, i.e.,
\begin{equation}
    P^{\textit{micro}} = P^{\textit{mid}} + g\left( q, S \right).
\end{equation}
We can then define a sequence of mid-price predictions as follows:
\begin{equation}
    P_t^i = \mathbb{E}[ P_{\tau_i}^{\textit{mid}} \given \mathcal{F}_t ],
\end{equation}
where $\mathcal{F}_t$ represents the information contained in the LOB at time $t$, and $\tau_i$ are the stopping times corresponding to the mid-price changes. By construction, the processes $P_t^i$ are martingales up to time $\tau_i$, and if the sequence $\{ P_t^i \}_i$ converges, it is also a martinagle.

Now, assuming that the information $\mathcal{F}_t$ is given by the 3-dimensional Markov process $\left( P_t^{\textit{mid}}, q_t, S_t \right)$ and that the mid-price changes are independent of $P_t^{\textit{mid}}$, one can prove that
\begin{equation}
    P_t^i = P_t^{\textit{mid}} + \sum_{k=1}^i g^k\left( q_t, S_t\right),
\end{equation}
where
\begin{equation}
    g^1\left( q, S \right) = \mathbb{E}[ P_{\tau_1}^{\textit{mid}} - P_{t}^{\textit{mid}} \given q_t=q, S_t=S ],
\end{equation}
and
\begin{equation}
    g^{i+1}\left( q, S \right) = \mathbb{E}[ g^i(q_{\tau_1}, S_{\tau_1}) \given q_t=q, S_t=S ]
\end{equation}
can be computed recursively.

The estimation of the micro-price employs spread and imbalance discretization as well as data symmetrization to ensure convergence. For further details, we refer the reader to \cite{Stoikov_2018, HAGSTROMER2021}. In our analysis, we used the implementation provided in\footnote{\label{sstoikov_github}\href{https://github.com/sstoikov/microprice}{https://github.com/sstoikov/microprice}} and\footnote{\label{micro_price_medium_github}\href{https://medium.com/@mhfizt/high-frequency-estimator-of-future-prices-micro-price-paper-code-walkthrough-475adb98e91d}{https://medium.com/@mhfizt/high-frequency-estimator-of-future-prices-micro-price-paper-code-walkthrough-475adb98e91d}}.

\subsection{Large-Tick Stock Example} 

By slightly modifying the example from \cite{Stoikov_2018}, we can establish a correspondence between the micro-price and the Boltzmann price (for $\beta \in (0,2]$). Assume the spread $P_t^a-P_t^b = 1$ (tick). Rather than changing by 1 tick, we assume that the mid-price $M_t = P_t^{\textit{mid}}$ can change by half a tick.\footnote{Alternatively, one can assume, for example, that the spread is constant at $P_t^a - P_t^b = 2$ ticks and the mid-price moves by 1 tick. This yields the same formula.} Such behavior is common for large-cap stocks, where the mid-price frequently jumps by $\frac{1}{2}$ tick. This occurs, for example, when either the bid or the ask moves by one tick, causing the mid-price to change by half a tick.\footnote{However, this example might be somewhat unrealistic because the mid-price moves by half a tick despite a constant spread equal to one tick. Alternatively, we may imagine that both the bid and ask prices also move by half a tick while maintaining a constant spread of one tick. We should treat this example as an idealization, although it illustrates a relationship that can be observed in practice.}

We model the mid-price movement by the process $I_t$ starting from time $t=0$. The process begins at $I_0 = I$
and evolves as a Brownian motion until hitting either 0 or 1. At the stopping time, the process instantaneously perturbs to $\epsilon$ or $1 - \epsilon$ with equal probability, then restarts as a new Brownian motion from that point.

We define the sequence of independent perturbations as a sequence of i.i.d. random variables $(X_k)_{k \geq 1}$ such that $\mathbb{P}(X_k=\epsilon) = \alpha$ and $\mathbf{P}(X_k=1-\epsilon) = 1 - \alpha$, where $\epsilon \in (0,\frac{1}{2})$ and $\alpha \in [0,1]$.

Let $\tau_0=0$ and
\begin{equation}
    \tau_k = \inf \lbrace{ t \geq \tau_{k-1} \colon I_t = 0 \; \textit{or} \; I_t = 1 \rbrace}.
\end{equation}
The process $I_t$ can be represented as a sequence of independent Brownian motions $(B^k)$ defined on consecutive time intervals $[\tau_{k}, \tau_{k+1})$, where $I_{\tau_k}=X_k$ at the stopping times. Then, for $t \in [\tau_{k}, \tau_{k+1})$,
\begin{equation}
    I_t = I_{\tau_{k}} + B_{t-\tau_{k}}^{(k+1)}.
\end{equation}
Notice that $I_{\tau_0^{-}} = I$, $I_{\tau_k^{-}} \in \lbrace{ 0, 1 \rbrace}$, and $I_{\tau_k} = X_k$.
We set the following conditional probabilities of mid-price jumps:\footnote{One can generalize this example by assigning different probabilities to upward or downward price movements. However, since the micro-price implementation uses data symmetrization, this generalization is not of relevant for our analysis.}
\begin{equation*}
    \mathbb{P}( M_{\tau_k} - M_{\tau_{k}^{-}} = 0.5, \: I_{\tau_k}=\epsilon \given I_{\tau_k^{-}} = 1 )= \alpha \quad \text{and} \quad \mathbb{P}( M_{\tau_k} - M_{\tau_k^{-}} = 0, \: I_{\tau_k}=1-\epsilon \given I_{\tau_{k}^{-}} = 1 ) = 1 - \alpha,
\end{equation*}
and
\begin{equation*}
    \mathbb{P}( M_{\tau_k} - M_{\tau_{k}^{-}} = - 0.5, \: I_{\tau_k}=1-\epsilon \given I_{\tau_{k}^{-}} = 0 ) = \alpha \quad \text{and} \quad \mathbb{P}( M_{\tau_k} - M_{\tau_k}^{-} = 0, \: I_{\tau_k}=\epsilon \given I_{\tau_{k}^{-}} = 0 ) = 1 - \alpha.
\end{equation*}
Then, the first adjustment of the micro-price is given by:
\begin{equation*}
    g^{1}(I,S) = g^{1}(I) = \mathbb{E}\left[M_{\tau_1} - M_0 \given I \right] = \frac{\alpha}{2} \mathbb{P}\left( I_{\tau_{1} } = 1 \given I \right) - \frac{\alpha}{2} \mathbb{P}\left(I_{\tau_{1} } = 0 \given I \right) = 
\end{equation*}
\begin{equation*}
    \alpha \left( I - \frac{1}{2} \right) = \alpha \left( P^a - P^b \right) \cdot \left( q^b - \frac{1}{2} \right),
\end{equation*}
where $I = q^b$. Hence, the micro-price is equal to:
\begin{equation*}
    P_{\textit{micro}} = \frac{P^b + P^a}{2} + \alpha \left( P^a - P^b \right) \cdot \left( q^b - \frac{1}{2} \right) = \frac{P^b + P^a}{2} + \alpha \cdot \left[ P^a \left( q^b - \frac{1}{2} \right) + P^b \left( q^a - \frac{1}{2} \right) \right] =
\end{equation*}
\begin{equation*}
    \frac{P^a \left( 1 - \alpha + 2\alpha q^b \right) + P^b \left( 1 - \alpha +2\alpha q^a \right)}{2} = (1 - \alpha) \cdot \frac{P^b + P^a}{2} + \alpha \cdot \left( P^a q^b + P^b q^a \right) = \left( 1 - \alpha \right) P^{\textit{mid}} + \alpha P^{\textit{w}}.
\end{equation*}
As in \cite{Stoikov_2018}, the subsequent adjustments satisfy $g^i\left( I_0, S_0\right) = 0$ for $i \geq 2$, and thus the micro-price is given by:
\begin{equation*}
    P^{\textit{micro}} = \left( 1 - \frac{\beta}{2} \right) P^{\textit{mid}} + \frac{\beta}{2} P^{\textit{w}} \approx P^{\textit{boltzmann}}(\beta),
\end{equation*}
where $\beta = 2\alpha \in [0,2]$, and the last approximation follows from Lemma~\ref{prop:price_approx}. In particular, for $\beta = 1$ (i.e., when the bid-ask bounce probability is equal to $\frac{1}{2}$), the micro-price is equal to:
\begin{equation}
    P^{\textit{micro}} = \frac{1}{2} \left( P^{\textit{mid}} + P_w \right) \approx P^{\textit{boltzmann}}\left( 1 \right) = P^{\textit{eq}}.
\end{equation}
In Figure~\ref{fig:micro_vs_eq}, we compare the micro-price with the Boltzmann price for $\beta = 1$. The dataset consists of one month of bid-ask quotes, with timestamp rounded to the nearest second, for two stocks with distinct microstructure characteristics: Bank of America (BAC) and Chevron (CVX). The data comes from \cref{sstoikov_github}.
We use the implementation from \cref{sstoikov_github} and \cref{micro_price_medium_github}, with 6 iterations (adjustments) and 10 imbalance states. For the spread, we apply 4 states for CVX and 2 for BAC.
We find that the first micro-price adjustment is typically the most significant. As noted in \cite{HAGSTROMER2021}, good convergence is usually achieved after 2-3 iterations.
In the case of BAC---where the spread is typically $0.01$ and the mid-price changes frequently by half the spread---the alignment between the micro-price and the equilibrium price is more apparent. For CVX, where the spread varies more, the relationship is less straightforward. However, as illustrated in Figure~\ref{fig:micro_vs_eq_b}, the two prices can still align at times, particularly when compared to the weighted mid-price.

\begin{figure}
\centering
\begin{subfigure}{0.5\textwidth}
  \includegraphics[scale=0.4]{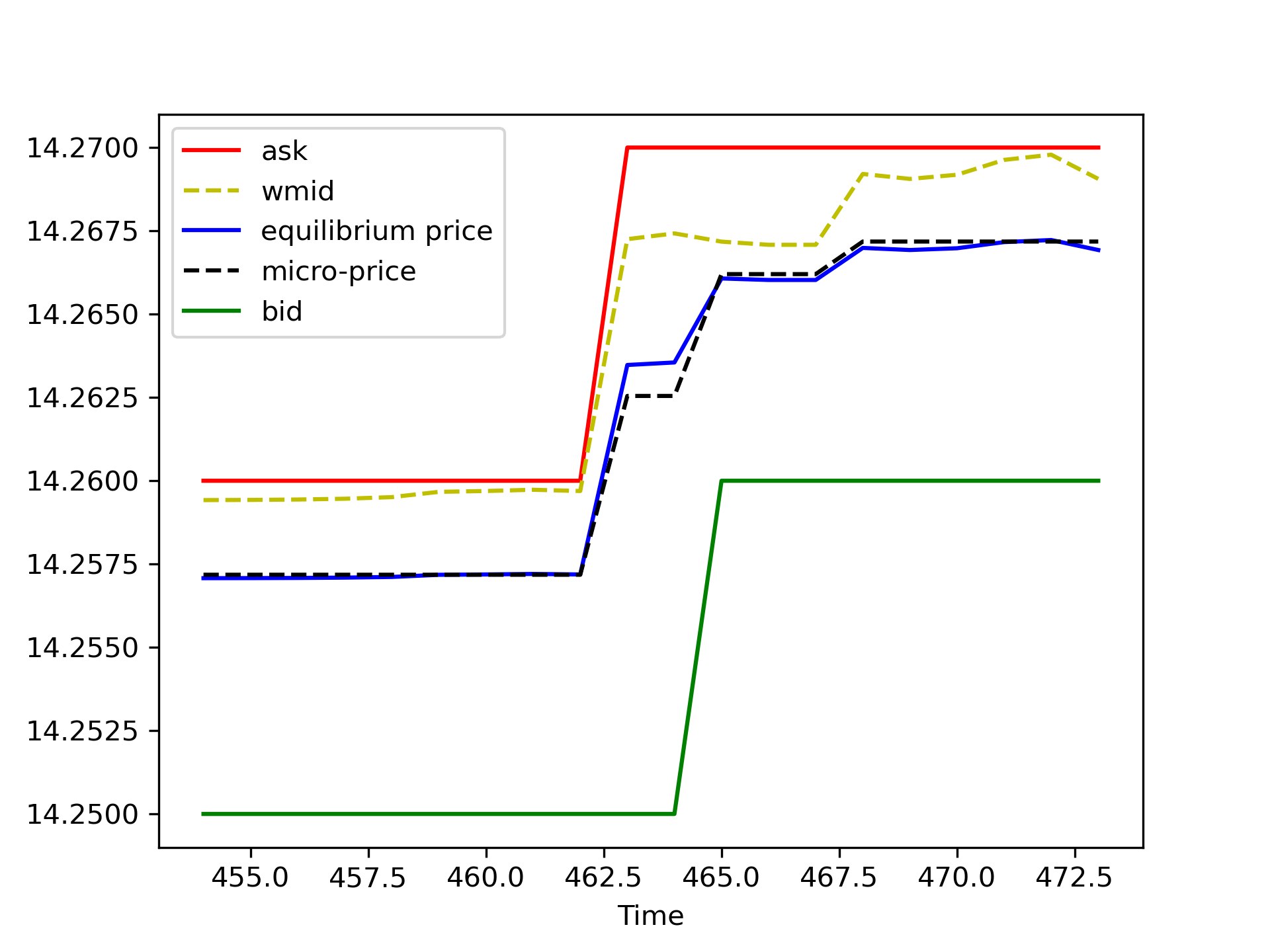}
  \caption{BAC}
  \label{fig:micro_vs_eq_a}
\end{subfigure}\hfill
\begin{subfigure}{0.5\textwidth}
  \includegraphics[scale=0.4]{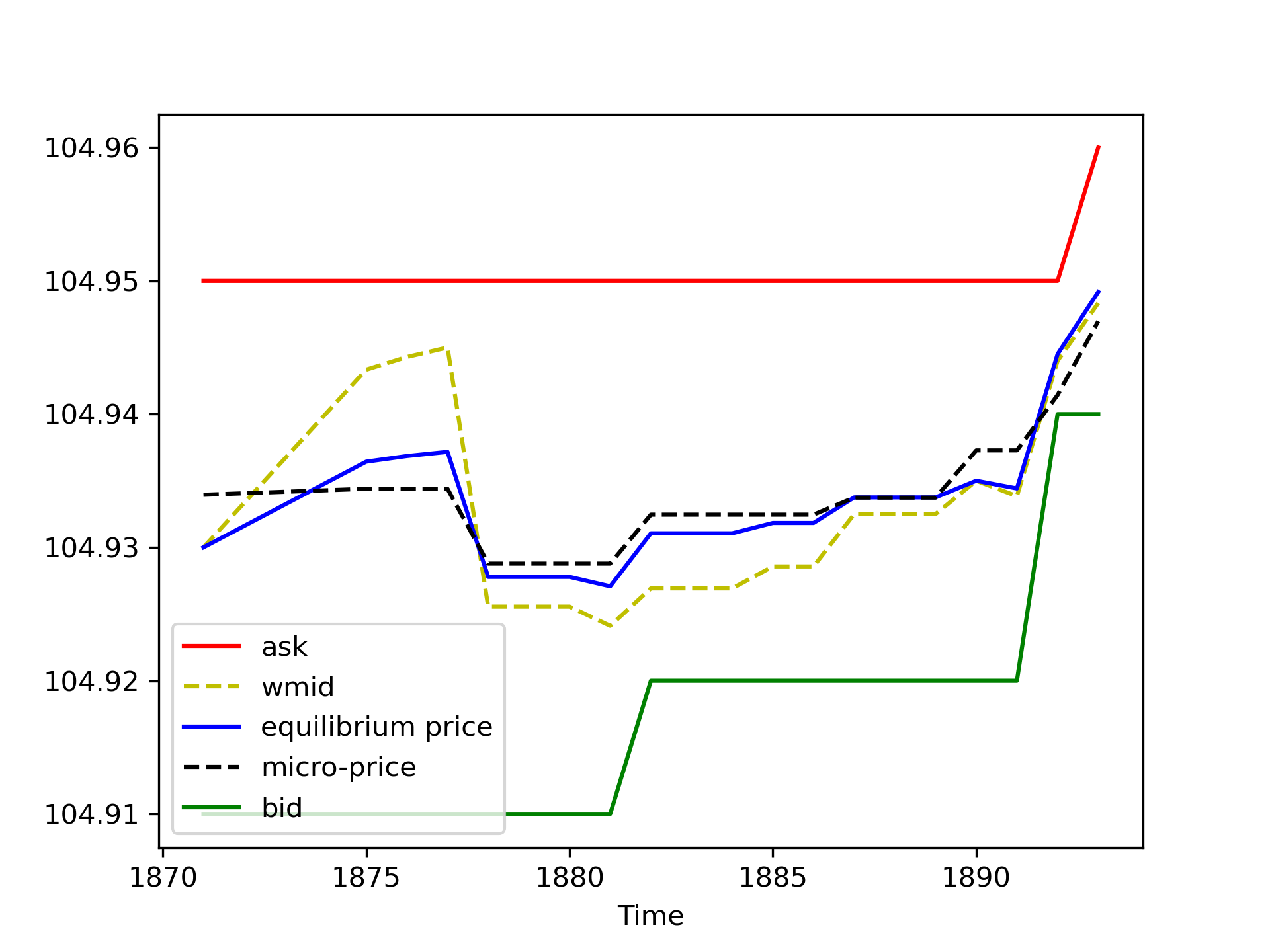}
  \caption{CVX}
  \label{fig:micro_vs_eq_b}
\end{subfigure}
\caption{Comparison of the micro-price and the equilibrium price dynamics over a short time interval.}
\label{fig:micro_vs_eq}
\end{figure}





\subsection{Relationship with the Spread}

In the previous example, the spread was 1 tick. Now, assume the spread is $P_t^a - P_t^b = S > 1$ tick, but the mid-price can still move by $\frac{1}{2}$ tick. Then, the micro-price is given by:
\begin{equation*}
    P_{\textit{micro}} = \frac{P^b + P^a}{2} + \frac{\alpha}{S} \cdot S \cdot \left( q^b - \frac{1}{2} \right) = \frac{P^b + P^a}{2} + \frac{\alpha}{S} \cdot \left[ P^a \left( q^b - \frac{1}{2} \right) + P^b \left( q^a - \frac{1}{2} \right) \right]
\end{equation*}
and using the same calculation, we obtain:
\begin{equation*}
    P_{\textit{micro}} = \left( 1 - \frac{\alpha}{S} \right) P^{\textit{mid}} + \frac{\alpha}{S} P^{\textit{w}} \approx P^{\textit{boltzmann}}\left(\frac{2\alpha}{S}\right) = P^{\textit{boltzmann}}\left(\frac{\beta}{S}\right).
\end{equation*}
However, when the spread exceeds one tick, the mid-price typically jumps by amounts larger than half a tick---often one tick or more. To model this behavior accurately, the conditional probabilities would need to be adjusted accordingly. Although it is indeed possible to generalize the model in this way, that is beyond the scope of this example. Instead, we focus on determining how the parameter $\beta$ is adjusted in relation to the spread.

Let us now calculate the ratio of the probabilities of two states, $P^b$ and $P^a$. We have:
\begin{equation*}
    \frac{p_b}{p_a} = \frac{\frac{1}{Z}e^{-\beta q^b}}{\frac{1}{Z} e^{-\beta q^a}} = e^{-\beta \left( q^b - q^a \right)},
\end{equation*}
where $Z= e^{-\beta q^b} + e^{-\beta q^a}$ is the so-called partition function. Hence,
\begin{equation}
    \beta = - \frac{1}{q^b - q^a} \ln \frac{p_b}{p_a}.
\end{equation}
We observe that when the spread is large, the TOB imbalance should have less impact on the price, and both states become approximately equally likely ($p_b \approx p_a$).\footnote{An analogy can be drawn with statistical mechanics: when the system is very hot, the probabilities of different states become similar.} Therefore, we obtain the following relation:
\begin{equation}
    \textit{spread} \to \infty \quad \implies \quad \beta \to 0 \quad \implies P^{\textit{boltzmann}} \to P^{\textit{mid}}.
\end{equation}
This means that for a large spread, the Boltzmann price should approximate the mid-price.

On the other hand, when the spread is small, the imbalance should have a significant impact on the price. In such cases, either $p_b \ll p_a$ or $p_a \ll p_b$, and in both cases $\beta \to \infty$. Thus, the relationship becomes inverted, and we have:
\begin{equation}
    \textit{spread} \to 0 \quad \implies \quad \beta \to \infty.
\end{equation}
As $\beta \to \infty$, the Boltzmann price $P^{\textit{boltzmann}}$ approaches $P^b$ when $q^b < q^a$, or $P^a$ when $q^b > q^a$. However, when the spread is very small, we have $P^b \approx P^a$, so the difference becomes negligible.

In summary, based on this heuristic reasoning, we expect the following relationship:
\begin{equation}
    \beta \sim \frac{1}{\textit{spread}}.
\end{equation}
Let $S = P_a - P_b$, and denote by $\tilde{p}_b$ and $\tilde{p}_a$ the probabilities of the bid and ask states, respectively, in the model with spread. Then
\begin{align}
    \tilde{p}_b &= \frac{e^{-\frac{\beta q^b}{S}}}{e^{-\frac{\beta q^b}{S}} + e^{-\frac{\beta q^a}{S}}}, \\
    \tilde{p}_a &= \frac{e^{-\frac{\beta q^a}{S}}}{e^{-\frac{\beta q^b}{S}} + e^{-\frac{\beta q^a}{S}}}.
\end{align}
It is therefore justified to introduce a price that accounts for the spread.

\begin{definition}
For $\beta \geq 0$, the generalized Boltzmann price is given by:
    \begin{equation}
        \tilde{P}\left( \beta \right) = \frac{e^{-\frac{\beta q^b}{S}} \cdot P^b + e^{-\frac{\beta q^a}{S}} \cdot P^a}{e^{-\frac{\beta q^b}{S}} + e^{-\frac{\beta q^a}{S}}} = \textit{softmax}\left( -\frac{\beta q^b}{S}, -\frac{\beta q^a}{S} \right) \cdot \left( P^b, P^a \right).
    \end{equation}
\end{definition}

One can observe that using the generalized Boltzmann price requires adjusting the parameter $\beta$ to fit the data. In particular, if we set $\beta = 1$, then for a large spread the generalized price will be close to the mid-price, according to Lemma~\ref{prop:price_approx}.

An analogy can be drawn between the kinetic theory of gases and our model. While the bid and ask volume imbalances $q^b$ and $q^a$ play the role of the energy states, the spread $S$ can be interpreted as the system's temperature.

\subsection{Spread and Imbalance from the Market Maker's Perspective}

Consider the profit and loss (P\&L) of market maker as described in \cite{Bouchaud_2009}.

Let $v_n \equiv v$ be the volume of the $n$-th trade\footnote{For simplicity, let us assume it is constant.} and
$\epsilon_n = 1$ indicates buy order, $\epsilon_n = -1$ the sell order. Then, the P\&L made by the liquidity provider in that given trade, mark-to-market at time $n+\ell$ is equal to \cite{Bouchaud_2009}:
\begin{equation}
    \mathcal{G}_L^{\textit{mid}}\left(n, n+\ell \right) = v \epsilon_n \left[ \left( m_n + \epsilon_n \frac{S}{2} - m_{n+\ell} \right) \right].
\end{equation}
On the other hand, the liquidity taker P\&L is $-\mathcal{G}_L^{\textit{mid}}\left(n, n+\ell \right)$.
Consider the liquidity taker's perspective. The average theoretical gain of the market maker is given by\footnote{We use the notation $\mathbb{E}[ \mathcal{G}_L^{\textit{mid}}]\left( \ell \right)$ to emphasize that the P\&L is compared to the mark-to-market mid price.}:
\begin{equation}
\mathbb{E}[ \mathcal{G}_L^{\textit{mid}}]\left( \ell \right) = v \left( \mathbb{E}\left[ \frac{S}{2} \right] - \mathcal{R}_{\ell} \right),
\end{equation}
where $\mathcal{R}_{\ell}$ is the lagged impact function, i.e.,
\begin{equation}
    \mathcal{R}_{\ell} = \mathbb{E}\left[ \epsilon_n \cdot \left( m_{n+\ell} - m_{n} \right) \right].
\end{equation}
How would the P\&L compare to the mark-to-market Boltzmann price? Similarly to above, we have:
\begin{equation}
    \mathcal{G}_L^{\textit{boltzmann}}\left(n, n + \ell \right) = v \epsilon_n \left[ \left( m_n + \epsilon_n \frac{S_n}{2} \right) - P_{n+\ell}^{\textit{boltzmann}} \right]
\end{equation}
Using Lemma \ref{prop:price_approx}, we have:
\begin{equation}
    \mathcal{G}_L^{\textit{boltzmann}}\left(n, n + \ell \right) \approx v \epsilon_n \left[ \left( m_n + \epsilon_n \frac{S_n}{2} \right) - m_{n+\ell} - \beta \frac{S_{n+\ell}}{2}\left( q_{n+\ell} - \frac{1}{2} \right) \right].
\end{equation}

The distance between the two P\&Ls can be approximated as follows:
\begin{equation}
    \left| \mathbb{E}[ \mathcal{G}_L^{\textit{mid}}]\left( \ell \right) - \mathbb{E}[ \mathcal{G}_L^{\textit{boltzmann}}]\left( \ell \right) \right| \lesssim \beta \mathbb{E}\left( S_{n+\ell} \left| \mathcal{I}_{n+\ell} \right| \right).
\end{equation}
This implies that the market maker should consider not only the market impact, $\mathcal{R}_{\ell}$, which contributes to potential losses, $-\epsilon_n \left( m_{n+\ell} - m_n \right)$, and gains from half the spread $\frac{S}{2}$, but also manage inventory risk. One way to achieve this is by widening or tightening the spread.



Consider now the market impact calculated for the Boltzmann price, i.e.,

\begin{definition}
    The lagged Boltzmann impact function is defined as:
    \begin{equation}
        \mathcal{R}_{\ell}^{\textit{boltzmann}} = \mathbb{E}\left[ \epsilon_n \cdot \left( P_{n+\ell}^{\textit{boltzmann}} - P_{n}^{\textit{boltzmann}} \right) \right].
    \end{equation}
\end{definition}

Let us calculate,
\begin{equation}
    \epsilon_n \left( P_{n+\ell}^{\textit{boltzmann}} - P_{n}^{\textit{boltzmann}} \right) \approx \epsilon_n \left( m_{n+\ell} - m_{n} \right) + \beta \epsilon_n \left( \frac{S_{n+\ell}}{2} \left( q_{n+\ell} - \frac{1}{2} \right) - \frac{S_n}{2} \left( q_n - \frac{1}{2} \right) \right).
\end{equation}
Then, one obtains the following:
\begin{equation}
    \mathcal{R}_{\ell}^{\textit{boltzmann}} \approx \mathcal{R}_{\ell} + \beta \mathbb{E}\left[ \epsilon_n \Delta_{\ell} \left( S\mathcal{I}\right) \right],
\end{equation}
where the $\ell$-lagged delta is
\begin{equation}
    \Delta_{\ell} \left( S\mathcal{I} \right) = \frac{S_{n+\ell}}{2} \left( q_{n+\ell} - \frac{1}{2} \right) - \frac{S_n}{2} \left( q_n - \frac{1}{2} \right) = S_{n+\ell} \mathcal{I}_{n+\ell} - S_n \mathcal{I}_n.
\end{equation}
How should we interpret the additional term $\Delta_{\ell} \left( S\mathcal{I} \right)$? As discussed, this approach highlights that market makers must account not only for the market impact, $\mathcal{R}_{\ell}$, but also for inventory risk, for instance, by widening or tightening the spread.

Market makers typically widen the spread when there is a significant order book imbalance, as this increases their exposure to adverse price movements. A wider spread helps mitigate this risk. Conversely, when the order book is more balanced, they can tighten the spread to attract more trading volume.


\newpage

\subsection{Additional Histograms for GE and LCID}

\begin{figure}[htbp]
    \centering
    \subfloat[\centering Imbalance]{{\includegraphics[scale=0.4]{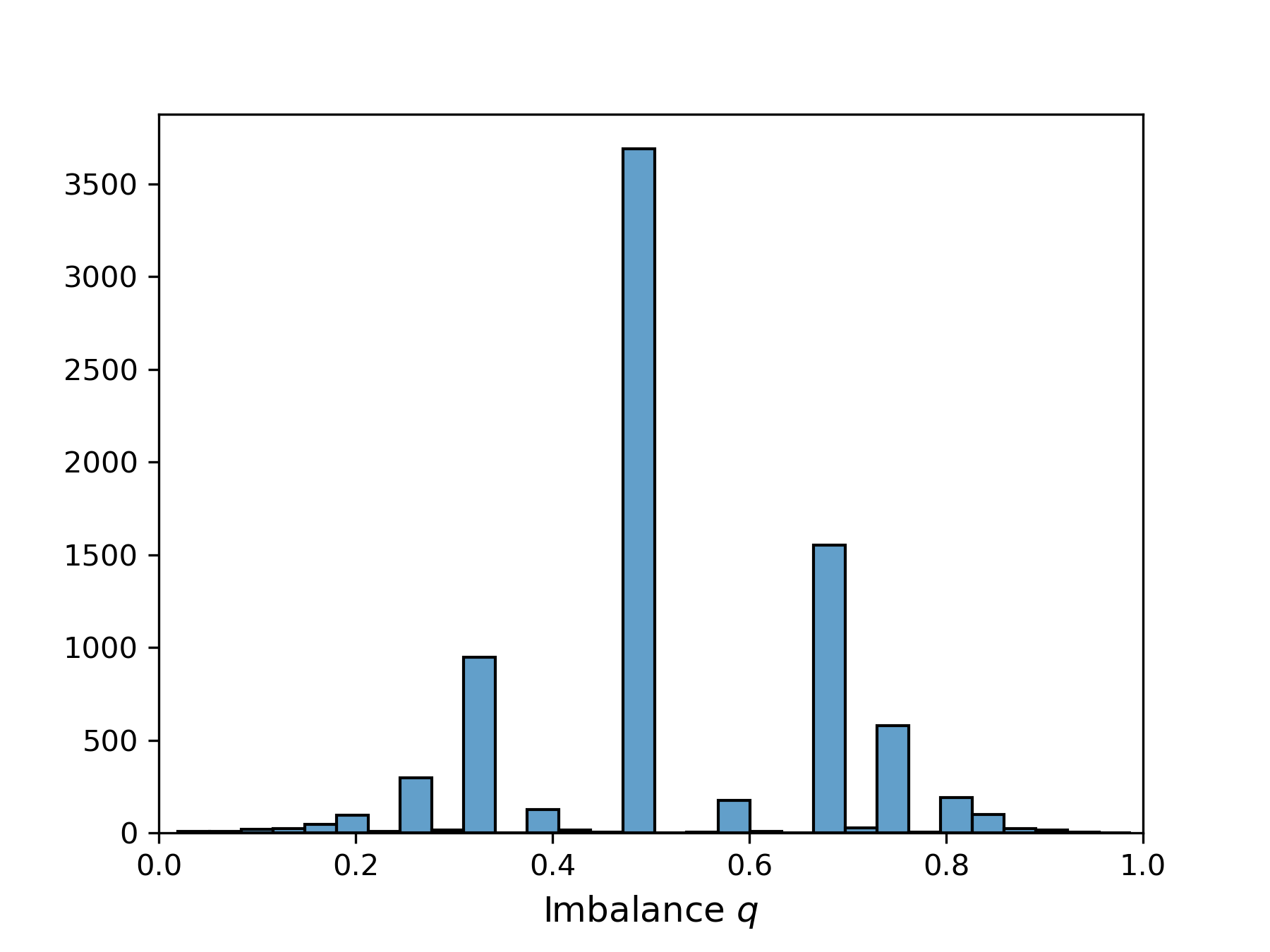} }}%
    \qquad
    \subfloat[\centering Sampled imbalance]{{\includegraphics[scale=0.4]{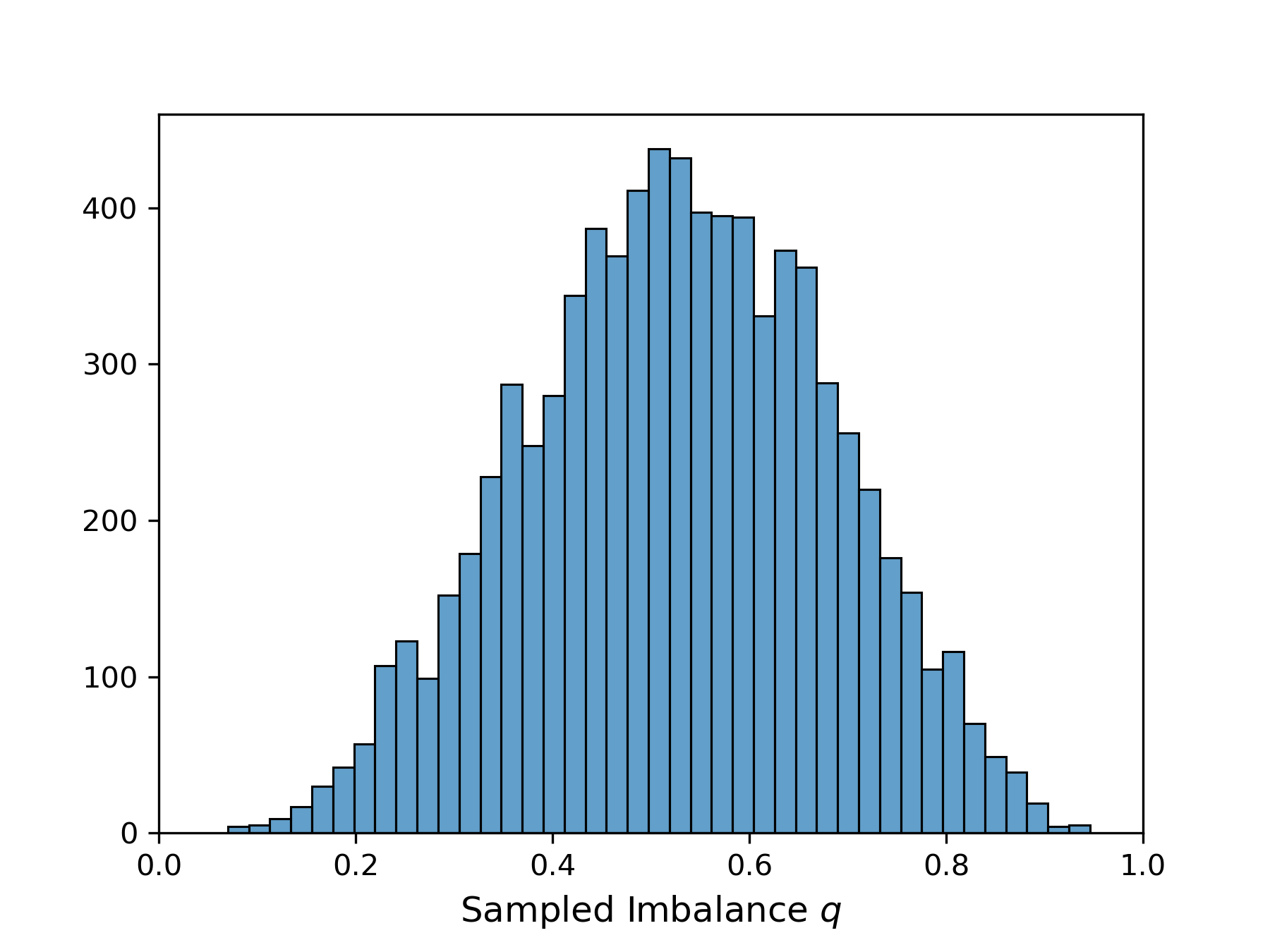} }}%
    \caption{Histograms of imbalance $q$ for GE, data is aggregated in $1$ min intervals.}%
    \label{fig:hist_imb_vs_sampled_imb_GE_N}
\end{figure}

\begin{figure}[htbp]
    \centering
    \subfloat[\centering Spread]{{\includegraphics[scale=0.4]{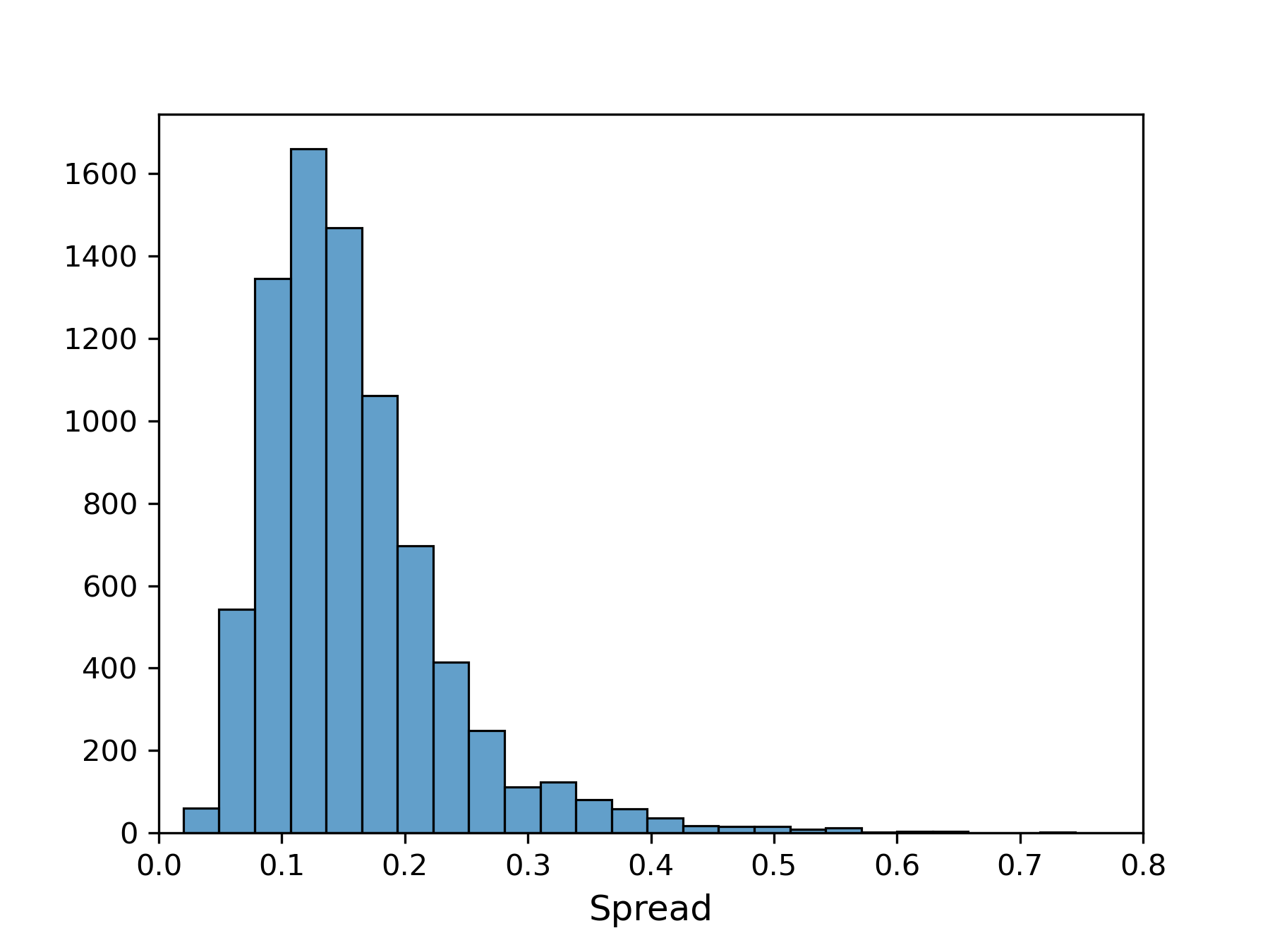} }}%
    \qquad
    \subfloat[\centering Sampled spread from $\Gamma(4.88,0.03)$.]{{\includegraphics[scale=0.4]{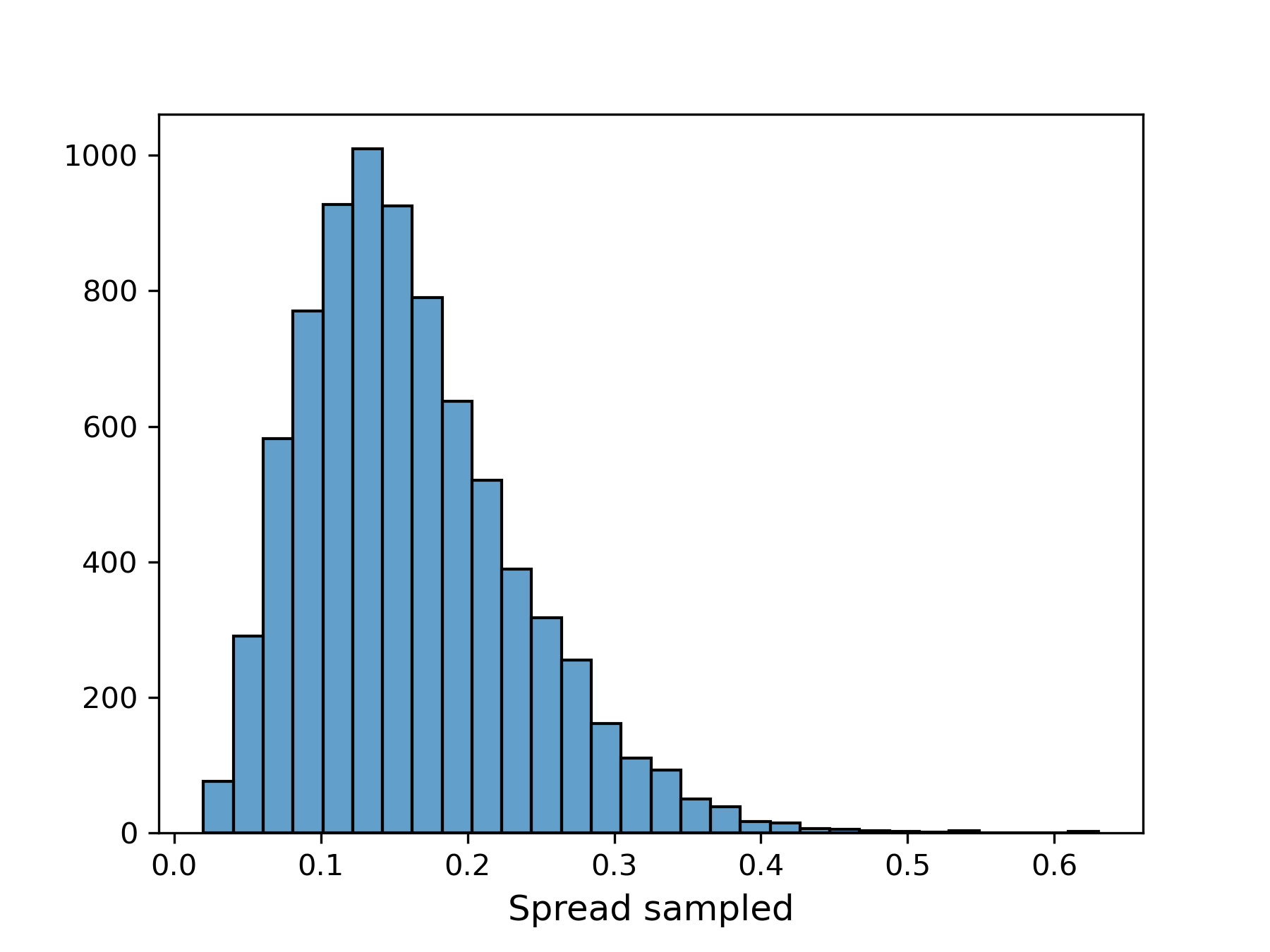} }}%
    \caption{Histograms of spread for GE, data is aggregated in $1$ min intervals.}%
    \label{fig:hist_spread_vs_sampled_spread_GE_N}
\end{figure}


\begin{figure}[htbp]
    \centering
    \subfloat[\centering Historical imbalance.]{{\includegraphics[scale=0.4]{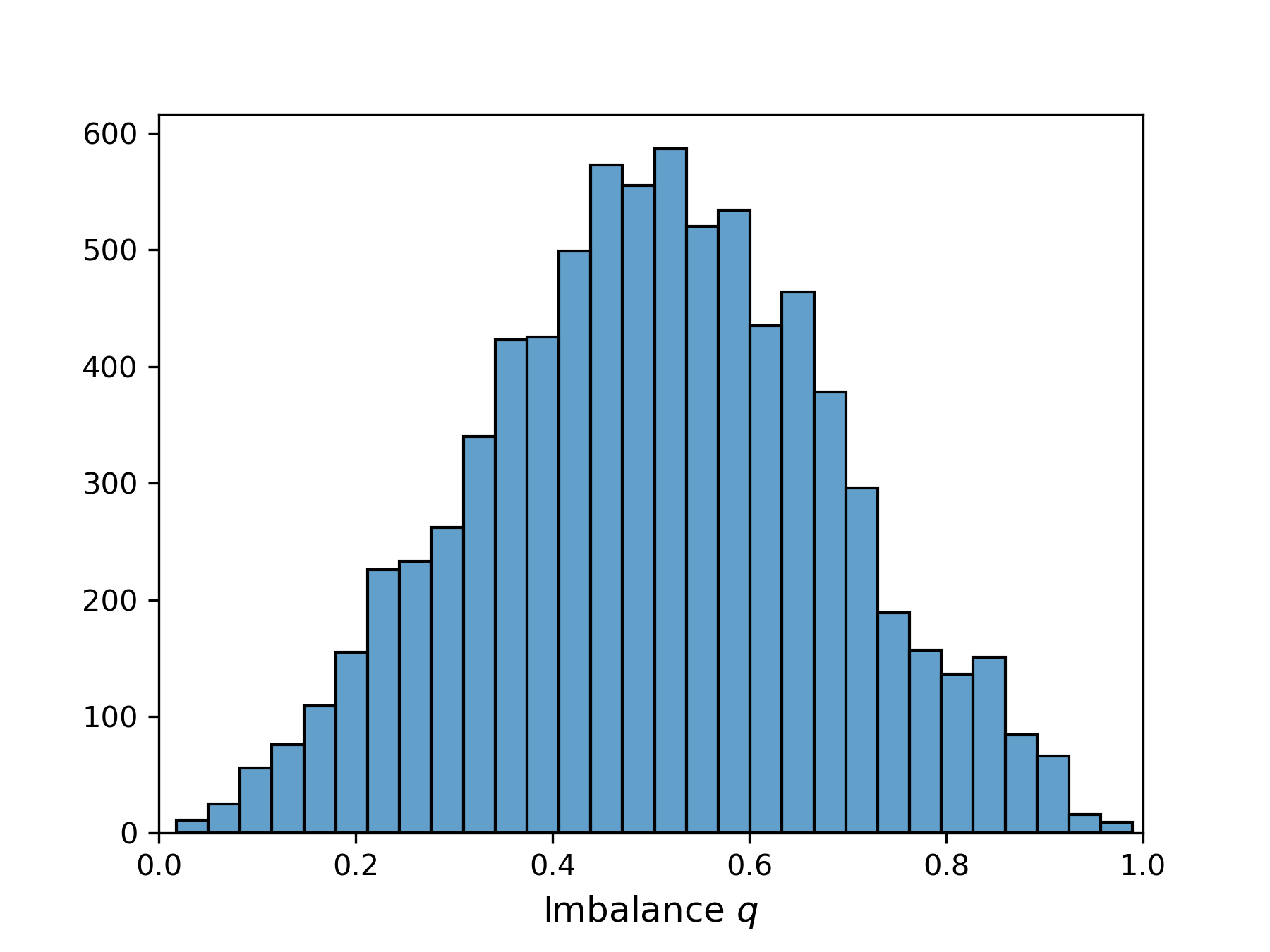} }}%
    \qquad
    \subfloat[\centering Sampled imbalance]{{\includegraphics[scale=0.4]{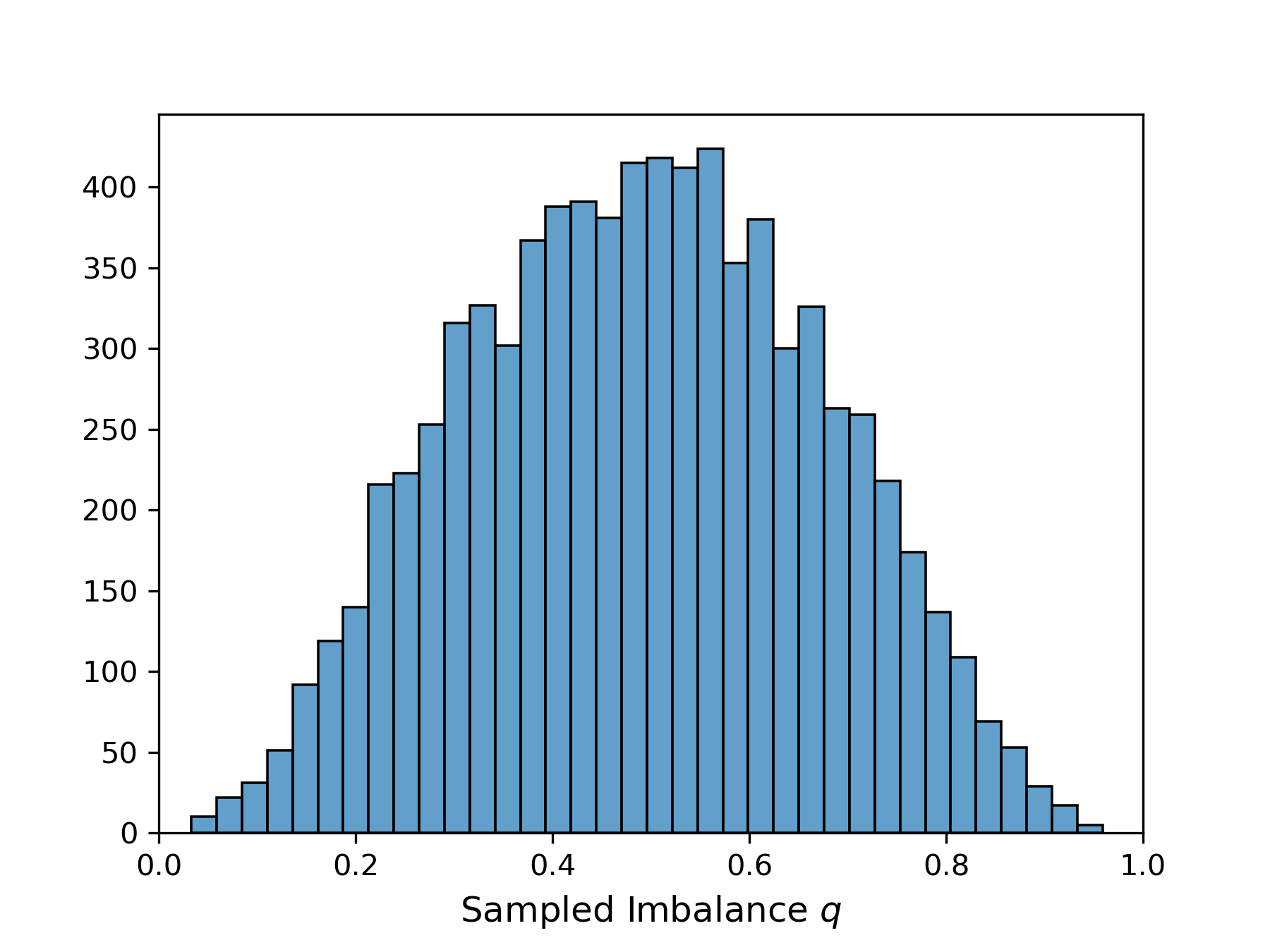} }}%
    \caption{Histograms of imbalance $q$ for LCID, data is aggregated in $1$ min intervals.}%
    \label{fig:hist_imb_vs_sampled_imb_LCID_O}
\end{figure}

\end{document}